\documentclass[11pt]{article}
\usepackage{fullpage}
\usepackage{amsmath,amsfonts,amsthm,amssymb}
\usepackage{url}

\usepackage{color}
\usepackage[usenames,dvipsnames,svgnames,table]{xcolor}
\usepackage[colorlinks=true, linkcolor=red, urlcolor=blue, citecolor=blue]{hyperref}

\usepackage{booktabs}
\usepackage{algorithm}
\usepackage{algpseudocode}
\usepackage{comment}
\usepackage{mathtools}
\usepackage{microtype}
\usepackage{caption}
\usepackage{float}
\usepackage{graphicx}
\usepackage{multirow}
\usepackage{mathrsfs}
\usepackage{newfloat}
\usepackage{relsize}

\numberwithin{equation}{section}
\numberwithin{figure}{section}
\theoremstyle{plain}
	\newtheorem{theorem}{Theorem}[section]

	\newtheorem{lemma}[theorem]{Lemma}

\theoremstyle{definition}
	\newtheorem{definition}[theorem]{Definition}
	
	\newtheorem*{remark*}{Remark}

\DeclareFloatingEnvironment[placement={!ht}]{myfloat}

\usepackage[framemethod=TikZ]{mdframed}
\newcounter{Frame}
\newenvironment{Frame}[1][]{%
\refstepcounter{Frame}
    \begin{mdframed}[%
        frametitle={#1},
        skipabove=\baselineskip plus 2pt minus 1pt,
        skipbelow=\baselineskip plus 2pt minus 1pt,
        linewidth=1.0pt,
        frametitlerule=true,
        nobreak=true
    ]%
}{%
    \end{mdframed}
}

\numberwithin{equation}{section}
\numberwithin{figure}{section}

\def\prob#1#2{\mbox{\bf Pr}_{#1}\left[ #2 \right]}

\def\expec#1#2{{\bf \mathbb{E}}_{#1}[ #2 ]}

\def\poly#1{{ \textrm{poly}}\left( #1 \right)}

\def\E{\mathbf{E}}

\author{
Ainesh Bakshi\thanks{Part of this work was done while Ainesh Bakshi and David Woodruff were visiting the Simons Institute for the Theorem of Computing. Ainesh Bakshi and David Woodruff acknowledge support in part from  NSF  No.  CCF-1815840.}
\\
Carnegie Mellon University\\
abakshi@cs.cmu.edu
\and
Nadiia Chepurko
\\
MIT\\
nadiia@mit.edu
\and David P. Woodruff\footnotemark[1]
\\
Carnegie Mellon University\\
dwoodruf@cs.cmu.edu}

\title{Weighted Maximum Independent Set of Geometric Objects in Turnstile Streams}
\date{}

\begin{document}

\maketitle

\begin{abstract}

We study the Maximum Independent Set problem for geometric objects given in the data stream model. A set of geometric objects is said to be independent if the objects are pairwise disjoint. We consider geometric objects in one and two dimensions, i.e., intervals and disks. Let $\alpha$ be the cardinality of the largest independent set. Our goal is to estimate $\alpha$ in a small amount of space, given that the input is received as a one-pass stream. We also consider a generalization of this problem by assigning weights to each object and estimating $\beta$, the largest value of a weighted independent set. We initialize the study of this problem in the turnstile streaming model (insertions and deletions) and provide the first algorithms for estimating $\alpha$ and $\beta$. 
	
For unit-length intervals, we obtain a $(2+\epsilon)$-approximation to $\alpha$ and $\beta$ in poly$(\frac{\log(n)}{\epsilon})$ space. We also show a matching lower bound. Combined with the $3/2$-approximation for insertion-only streams by Cabello and Perez-Lanterno \cite{cabello2017interval}, our result implies a separation between the insertion-only and turnstile model. For unit-radius disks, we obtain a $\left(\frac{8\sqrt{3}}{\pi}\right)$-approximation to $\alpha$ and $\beta$ in $\poly{\frac{\log(n)}{\epsilon}}$ space, which is closely related to the hexagonal circle packing constant.

%For arbitrary-length intervals, we show any $c$-approximation to $\alpha$, and thus also $\beta$, requires $\Omega\left(\frac{n^{1/c}}{2^c}\right)$ space. To this end, we introduce a new communication problem and lower bound its information complexity.  
%To cope with the lower bound, 
We provide algorithms for estimating $\alpha$ for arbitrary-length intervals under a bounded intersection assumption and study the parameterized space complexity of estimating $\alpha$ and $\beta$, where the parameter is the ratio of maximum to minimum interval length.

\end{abstract}

\clearpage

\section{Introduction} 
Maximum Independent Set (\textsf{MIS}) is a fundamental combinatorial problem and in general, is NP-Hard to approximate within a $n^{1-\epsilon}$ factor, for any constant $\epsilon>0$  \cite{Has96}. %Poly-Apx-Complete, i.e., no approximation factor is possible within $|V|^{1-\epsilon}$ for $\epsilon > 0$. 
We focus on the \textsf{MIS} problem for geometric objects: we are given as input $n$ intervals on the real line or disks in the plane and our goal is to output the largest set of non-overlapping intervals or disks. 
Computing the Maximum Independent Set of intervals and disks has numerous applications in scheduling, resource allocation, cellular networks, map labellings, clustering, wireless ad-hoc networks and coding theory,  where it has been extensively studied \cite{gavril1972algorithms} \cite{awerbuch1994competitive}, \cite{woeginger1994line}, \cite{canetti1998bounding}, \cite{bachmann2010online}, \cite{azar2015buffer}, \cite{agarwal1998label}, \cite{hale1980frequency}, \cite{malesinska1997graph}. 
%We are given a set of $n$ intervals and the goal is to find the largest set of non-intersecting intervals. We can view the Interval Selection problem as computing the Maximum Independent Set of the induced Intersection Graph, formed by assigning a vertex to each interval and creating edge when two intervals intersect. 

%\textbf{Intersection Graphs.} An intersection graph, $G(V, E)$ is defined to be a graph whose vertex set $V$ is represented by having a node corresponding to each object, and two nodes share an edge if and only if the corresponding objects have a non-zero intersection, see, e.g, \cite{EGP66}. 
%Intersection Graphs have been studied widely in the past \cite{EGP66}. 
%We are interested in studying intersection graphs of geometric objects, such as intervals and disks, in the streaming model. 
%Prominent applications of intersection graphs include frequency assignment in cellular networks, map labeling, clustering, wireless ad-hoc networks and coding theory.

%Given an intersection graph, some important properties to study include the maximum independent set (\textsf{MIS}), vertex cover, and maximum clique. 

In the one dimensional setting, the \textsf{MIS} problem, also known as the Interval Scheduling\footnote{See \url{https://en.wikipedia.org/wiki/Interval_scheduling}} problem, has a simple greedy algorithm %receives as input a set of $n$ intervals in $\mathbb{R}$ and the goal is to find the largest cardinality set of non-intersecting intervals.  We also consider a natural generalization of this problem, Weighted Maximum Independent Set  (\textsf{WMIS}), where the input is now a set of weighted intervals and the goal is to find a set of non-intersecting intervals with maximum weight. 
%The classical algorithm of \cite{gavril1972algorithms} for \textsf{MIS} 
that picks intervals in increasing order of their right endpoint to obtain an optimal solution. The variant with weighted intervals can also be solved in polynomial time using dynamic programming, which is shown in a number of modern algorithms textbooks \cite{cormen2009introduction}, \cite{kleinberg2006algorithm}. These algorithms have considerable applications in resource allocation and scheduling, where offline and online variants have been extensively studied and we refer the reader to \cite{kolen2007interval} for a survey.

In the two dimensional setting, \textsf{MIS} of geometric objects, such as line segments \cite{hlinveny2001contact}, rectangles \cite{fowler1981optimal}, \cite{imai1983finding} and disks \cite{clark1990unit}, is \emph{NP-Hard}. However, in the offline setting (polynomial space), a PTAS is known for fat objects (squares, disks) and pseudo-disks \cite{ChanH12} (who also provide a recent survey). The \textsf{MIS} problem for arbitrary rectangles has also received considerable attention: \cite{chalermsook2009maximum} show a $\log\log(n)$ approximation in polynomial time and \cite{chuzhoy2016approximating} obtain a $(1+\epsilon)$-approximation in $n^{\poly{\log(n)}\epsilon^{-1}}$ time for axis-aligned rectangles.

\textbf{Streaming Model.} 
The increase in modern computational power has led to massive amounts of available data. Therefore, it is unrealistic to assume that our data fits in RAM. Instead, working with the assumption that data can be efficiently accessed in a sequential manner has led to streaming algorithms for a number of problems. Several classical problems such as heavy-hitters and $l_p$ sampling \cite{jowhari2011tight}, $l_p$ estimation \cite{kane2010exact}, entropy estimation \cite{li2011new}, \cite{clifford2013simple}, maximum matching \cite{konrad2015maximum} etc. have been studied in the turnstile model and recent work has led to interesting connections with linear sketches \cite{AHLW16}.

In this paper, we study the streaming complexity of the geometric \textsf{MIS} problem, where the input is a sequence of $n$ updates, either inserting a new object or deleting a previously inserted object. We assume that the algorithm has poly-logarithmic bounded memory and at the end of the stream, the algorithm should output an estimate of the (weighted) cardinality of the \textsf{MIS}. Since most real world scheduling applications are dynamic, and scheduling constraints expire, it is crucial to allow for both insertions and deletions, while operating in the low-space setting. Consider the following concrete application: automatic point-label conflict resolution on interactive maps \cite{mote2007fast}. In this problem, the goal is to label features (geometric objects such as points, lines and polygons) on a map such that no two features with the same label overlap. Labelling maps in visual analytic software requires such labelling to be fast and dynamic, since features can be added and removed.

%We extend this line of work to include \textsf{WMIS} of geometric objects in turnstile streams. 

%Motivated by this and the numerous applications of Interval Graphs, we study the \textsf{MIS} and \textsf{WMIS} problems in the single-pass turnstile streaming model. Here each input is an interval, which we get to see exactly once and the space available is sublinear (typically poly$(\log(n))$) in the size of the input. In this model, we may see both insertions of intervals and deletions of intervals we have already seen. 

\begin{center}
\begin{table}
\begin{center}
\begin{tabular}{ |c|c|c|c|c| }
\hline
Problem & \multicolumn{2}{c|}{Insertion-Only Streams} &  \multicolumn{2}{c|}{Turnstile Streams}\\ \hline
 & upper bound & lower bound &upper bound & lower bound  \\

\hline
Unit Intervals& $3/2+\epsilon$ & $3/2-\epsilon$ & $2+\epsilon$ & $2-\epsilon$ \\
Unit Weight& \cite{CP15}&  \cite{CP15} & Thm \ref{thm:weighted_ins_del_unit}& Thm \ref{athm:ins_del_unit_lowerbound} \\
 
\hline
Unit Intervals & $3/2+\epsilon$ & $3/2-\epsilon$ & $2+\epsilon$ & $2-\epsilon$ \\
Arbitrary Weight&  Thm \ref{thm:weighted_ins_unit_int} & \cite{CP15}  & Thm \ref{thm:weighted_ins_del_unit}& Thm \ref{athm:ins_del_unit_lowerbound} \\

%\hline
%Arbitrary Intervals & $2+\epsilon$ & $2-\epsilon$ & %$O(\log(n))$ & $\Omega(\log(n))$  \\
%Unit Weight& \cite{CP15} &  \cite{CP15} & folklore & %Thm \ref{thm:lowerbound_arb_len_unit_wt_ins_del} \\
%\hline
%Unit Disks & $\frac{8\sqrt{3}}{\pi} +\epsilon$ & $2 - \epsilon$ & %$\frac{8\sqrt{3}}{\pi} +\epsilon$ & $2-\epsilon$ \\
%Unit Weight& Thm \ref{thm:turnstile_geometric_disks} & Thm \ref{thm:insertion_only_disks_lowerbound} &Thm \ref{thm:turnstile_geometric_disks}& Thm \ref{athm:ins_del_unit_lowerbound}  \\
\hline
Unit Disks & $\frac{8\sqrt{3}}{\pi} +\epsilon$ & $2 - \epsilon$ & $\frac{8\sqrt{3}}{\pi} +\epsilon$ & $2-\epsilon$ \\
Arbitrary Weight & Thm \ref{thm:turnstile_geometric_disks} & Thm \ref{thm:insertion_only_disks_lowerbound} &Thm \ref{thm:turnstile_geometric_disks}& Thm \ref{athm:ins_del_unit_lowerbound}  \\
%\hline
%\multirow{2}{*}{Arbitrary Length, Unit Weight Disks} &  &  & &   \\
%&  &   &  &  \\
\hline
\end{tabular}
\end{center}
\caption{The best known upper and lower bounds for estimating $\alpha$ and $\beta$ in insertion-only and turnstile streams (defined below). Note, the weight and length above are still polynomially bounded in $n$. The folklore result follows from partitioning the input into $O(\log(n))$ weight classes, estimating $\alpha$ on each one in parallel and taking the maximum estimate.}
\label{tab}
\end{table}
\end{center}

\subsection{Our Contributions}

We provide the first algorithmic and hardness results for the Weighted Maximum Independent Set (\textsf{WMIS}) problem for geometric objects in {\it turnstile streams} (where previously inserted objects may also be deleted). The aim of our work is to understand the \textsf{MIS} and \textsf{WMIS} problems in this common data stream model and we summarize the state of the art in Table \ref{tab}. 
%We also introduce the weighted analogue of these problems in the context of turnstile streams.
%We now present our results for the \textsf{MIS} and \textsf{WMIS} problems on Interval Intersection Graphs. 
Our contributions are as follows:
\begin{enumerate}
\item \textbf{Unit-length Intervals.} Our main algorithmic contribution is a turnstile streaming algorithm achieving a $(2+\epsilon)$-approximation to $\alpha$ and $\beta$ in poly$\left(\frac{\log(n)}{\epsilon}\right)$ space. We also show a matching lower bound, i.e., any (possibly randomized) algorithm approximating $\alpha$ up to a $(2-\epsilon)$ factor requires $\Omega(n)$ space. Interestingly, this shows a strict separation between insertion-only and turnstile models since \cite{CP15} show that a $3/2$ approximation is tight in the insertion-only model. 

An unintuitive yet  crucial message here is that attaching polynomially bounded weights to intervals does not affect the approximation factor. Along the way, we also obtain new algorithms for estimating $\beta$ in insertion-only streams which are presented in Section \ref{sec:ins_only}. 
%\item \textbf{Arbitrary-length Intervals.} Our main hardness result is a lower bound that rules out any constant-factor approximation to $\alpha$ in poly$\left(\frac{\log(n)}{\epsilon}\right)$ space for turnstile streams. In order to show the lower bound, we introduce a new communication problem called the Augmented Path Disjointness Problem (\textsf{APD}), which may be useful in other contexts. 

We study this problem in the two player one-way communication model and show that an algorithm for estimating $\alpha$ implies a protocol for \textsf{APD}. We show an $\Omega(n^{1/s}/2^s)$ communication lower bound for this problem, where $\Theta(s)$ is the desired approximation ratio for the streaming algorithm. We note that this also implies a lower bound for estimating $\beta$ for arbitrary-length intervals. 
\item \textbf{Arbitrary Length Intervals.} For arbitrary length intervals, we give a one-pass turnstile streaming algorithm that achieves a $(1+\epsilon)$-approximation to $\alpha$ under the assumption that the degree of the interval intersection graph is bounded by poly$\left(\frac{\log(n)}{\epsilon}\right)$. Our algorithm achieves poly$\left(\frac{\log(n)}{\epsilon}\right)$ space.
%\textbf{I am not sure how much detail should the techniques we use be introduced in}

We also study the problem for arbitrary lengths by parameterizing the ratio of the longest to the shortest interval. We give a one-pass turnstile streaming algorithm that achieves a $(2+\epsilon)$-approximation to $\alpha$, where the space complexity is parameterized by $W_{max}$, which is an upper bound on the length of an interval assuming the minimum interval length is $1$. Here, the space complexity of our algorithm is poly$\left(W_{max} \frac{\log(n)}{\epsilon}\right)$ and this algorithm gives sublinear space whenever $W_{max}$ is sublinear. 
%The techniques we use here are similar to the ones we introduce to estimate $\alpha$ under the bounded degree assumption.   
\item \textbf{Unit-radius Disks.} We show that we can extend the ideas developed for unit-length intervals in turnstile streams to unit disks in the 2-d plane. We describe an algorithm achieving an  $\left(\frac{8\sqrt{3}}{\pi}+\epsilon\right)$-approximation to $\alpha$ and $\beta$ in poly$\left(\frac{\log(n)}{\epsilon}\right)$ space. One key idea in the algorithm is to use the hexagonal circle packing for the plane, where the fraction of area covered is $\frac{\pi}{\sqrt{12}}$ and our approximation constant turns out to be $4\cdot \frac{\sqrt{12}}{\pi}$. 

We also show a lower bound that any (possibly randomized) algorithm approximating $\alpha$ or $\beta$ for disks in insertion-only streams, up to a $(2-\epsilon)$ factor requires $\Omega(n)$ space. This shows a strict separation between estimating intervals and disks in insertion-only streams. 
\end{enumerate}

\section{Related Work}

There has been considerable work on streaming algorithms for graph problems.
%, including the introduction of the semi-streaming model. 
Well-studied problems include finding sparsifiers, identifying connectivity structure, building spanning trees, and matchings; see the survey by McGregor \cite{Mcgregor2014graph}. 
%Additionally, there has been a line of work on estimating the cardinality of a maximum matching in a stream. Assadi et. al. \cite{Assadi2017estimating} show $\alpha$-approximating the cardinality of of the maximum matching in $O(\frac{n}{\alpha^2})$ space for insertion-only streams and $O(\frac{n^2}{\alpha^4})$ space for turnstile streams. For bounded arboricity graphs, there exist algorithms that approximate the cardinality of the maximum matching in $O(\log^2(n))$ space \cite{CormodeJMM17}. 
Recently, Cormode et. al. \cite{CormodeDK17} provide guarantees for estimating the cardinality of a maximum independent set of general graphs via 
the Caro-Wei bound.
Emek, Halldorsson and Rosen \cite{EHR12} studied estimating the cardinality of the maximum independent set for interval intersection graphs in insertion-only streams. They output an independent set that is a $\frac{3}{2}$-approximation to the optimal (\textsf{OPT}) for unit-length intervals and a $2$-approximation for arbitrary-length intervals in $O(|\text{\textsf{OPT}}|)$ space. Note that $|\text{\textsf{OPT}}|$ could be $\Theta(n)$ which is a prohibitive amount of space.

Subsequently, Cabello and Perez-Lantero \cite{CP15} studied the problem of estimating the \emph{cardinality} of \textsf{OPT}, which we denote by $\alpha$, for unit-length and arbitrary length intervals in one-pass insertion-only streams. 
For unit-length intervals in insertion-only streams, Cabello and Perez-Lantero \cite{CP15} give a $(\frac{3}{2} + \epsilon)$ approximation to $\alpha$ in poly$\left(\frac{\log(n)}{\epsilon}\right)$ space. Additionally, they show that this approximation factor is tight, since any algorithm achieving a $(\frac{3}{2} -\epsilon)$-approximation to $\alpha$ requires $\Omega(n)$ space.
For arbitrary-length intervals they give a $(2 + \epsilon)$-approximation to $\alpha$ in poly$\left(\frac{\log(n)}{\epsilon}\right)$ space. Additionally, they show that the approximation factor is tight, since any algorithm achieving a $(2 -\epsilon)$-approximation to $\alpha$ requires $\Omega(n)$ space. Recently, \cite{cormode2018independent} studied MIS of intersection graphs in insertion-only streams. They show achieving a $\left(5/2-\epsilon\right)$-approximation to MIS of squares requires 
$\Omega(n)$ space.

To the best of our knowledge there is no prior work on the problem of Maximum Independent Set of unit disks in turnstile streams. In the offline setting, the first PTAS for MIS of disks was developed by \cite{erlebach2005polynomial} and later improved in running time by Chan \cite{chan2003polynomial}, while \cite{hochbaum1985approximation} shows a PTAS for MIS of $k \times k$ squares. We note that these algorithms require space linear in the number of disks and use a dynamic programming approach that is not suitable for streaming scenarios.

We note that \textsf{MIS} can also be viewed as a natural generalization of the distinct elements problem that has received considerable attention in the streaming model. This problem was first studied in the seminal work of \cite{flajolet1985probabilistic} and a long sequence of work has addressed its space complexity in both insertion-only and turnstile streams  \cite{alon1996space},  \cite{bar2002counting},  \cite{gibbons2001estimating},  \cite{estan2003bitmap}, \cite{flajolet2007hyperloglog}, \cite{DBLP:conf/pods/KaneNW10}, \cite{blasiok2018optimal} and \cite{cormode2018independent}.

\section{Notation and Problem Definitions}
We let $D(d_j, r_j, w_j)$ be a disk in $\mathbb{R}^d$, where $d \in \{1,2\}$, such that it is centered at a point $d_j \in \mathbb{R}^d$ with radius $r_j \in \mathbb{N}$ and weight $w_j$.
We represent $D(d_j, r_j, w_j)$ using the short form $D_j$ when $d_j$, $r_j$ and $w_j$ are clear from context. Note, we use the same notation to denote intervals in $d = 1$.
For a set $\mathcal{P}\subseteq \mathbb{R}^{d}$ of $n$ disks 
(unweighted or weighted), let $G$ be the induced graph formed by assigning a vertex to each disk and adding an edge between two vertices if the corresponding disks intersect. We call $G$ an intersection graph. 
The Maximum Independent Set (\textsf{MIS}) and Weighted Maximum Independent Set (\textsf{WMIS}) problems in the context of intersection graphs are defined as follows: 
	
\begin{definition}[Maximum Independent Set]
Let $\mathcal{P} = \{ D_1, D_2 \ldots , D_n\} \subseteq \mathbb{R^d}$ be a set of $n$ disks such that each weight $w_j = 1$ for $j \in [n]$. 
The \textsf{MIS} problem is to find the largest disjoint subset  $\mathcal{S}$ of $\mathcal{P}$ (i.e., no two objects in $\mathcal{S}$ intersect). We denote the cardinality of this set by $\alpha$. 
\end{definition}

\begin{definition}[Weighted Maximum Independent Set]
Let $\mathcal{P} = \{ D_1, D_2 \ldots , D_n\} \subseteq \mathbb{R}^d$ be a set of $n$ weighted disks.  
We let the weight $w_\mathcal{S}$ of a subset $\mathcal{S} \subseteq \mathcal{P}$ be $w_\mathcal{S}=\sum_{D_j \in \mathcal{S}} w_{j}$. 
The \textsf{WMIS} Problem is to find a disjoint (i.e., non overlapping) subset $\mathcal{S}$ of $\mathcal{P}$ whose weight $w_\mathcal{S}$ is maximum. We denote the weight of the \textsf{WMIS} by $\beta$.
\end{definition}
\vspace{-0.05in}

For a set $\mathcal{P}$ of disks, let \textsf{OPT}$_\mathcal{P}$ denote \textsf{MIS} or \textsf{WMIS} of $\mathcal{P}$. We use $|\textsf{OPT}_\mathcal{P}|$ to denote the cardinality of \textsf{MIS} as well as the weight of \textsf{WMIS} for $\mathcal{P}$. When the set $\mathcal{P}$ is clear from context, we omit it. Next, we define the two streaming models we consider. In our context, an \textit{insertion-only} stream provides sequential access to the input, which is an ordered set of objects such that at any given time step a new interval arrives. \textit{Turnstile streams} are an extension of this model such that at any time step, previously inserted objects can be deleted. An algorithm in the streaming model has access to space sublinear in the size of the input and is restricted to making one pass over the input. 

For proving our lower bounds, we work in the 
two player one-way randomized communication complexity model, where 
the players are denoted by Alice and Bob, who have private randomness. The input of Alice is denoted by $X$ and the input for Bob is denoted by $Y$. The objective is for Alice to
communicate a message to Bob and compute a function $f: X \times Y \to \{0, 1\} $ on the joint inputs of the players. The 
communication is one-way and w.l.o.g. Alice sends one message to Bob and Bob outputs a bit denoting the answer to the communication problem. Let
$\Pi\left(X,Y\right)$ be the random variable that denotes the 
transcript between sent from Alice to Bob when they execute a protocol 
$\Pi$. 

A protocol $\Pi$ is called a $\delta$-error protocol for function $f$ if there 
exists a function $\Pi_{out}$ such that for every input 
$Pr\left[\Pi_{out}\left(\Pi(X,Y)\right) = f(X, Y)\right] \geq 1 - \delta$. The communication cost of a protocol, denoted by 
$|\Pi|$, is the maximum length of $\Pi\left(X,Y\right)$ over 
all possible inputs and random coin flips of the two players. The randomized communication complexity of a function $f$, 
$R_{\delta}(f)$, is the communication cost of the best $\delta$-error protocol 
for computing $f$. 

\section{Technical Overview}

In this section, we summarize our results and briefly describe the main technical ideas in our algorithms and lower bounds. We note that our results hold in the recently introduced Sketching Model \cite{sun2019querying}. This model captures applications of sketches in turnstile streams, distributed computing, communication complexity and property testing. While Sun et. al. study graph problems such as dynamic connectivity and triangle detection, we initiate the study of dynamic Maximum Independent Set in this model. While we state our results in for turnstile streams, they immediately extend to the sketching model.

\subsection{Unit-length Intervals}
%\vspace{0.1in}
%\noindent \textbf{.}
Our main algorithmic contribution is to provide an estimate that obtains a $(2+\epsilon)$-approximation to \textsf{WMIS} of unit-length intervals in turnstile streams : 

\begin{theorem}[Theorem \ref{thm:weighted_ins_del_unit}, informal]
For any $\epsilon >0$, there exists a turnstile streaming algorithm that outputs an estimate such that with probability at least $99/100$, it is a $(2+\epsilon)$-approximation to $\textsf{WMIS}$ of unit intervals (polynomially bounded weights) and the algorithm requires  poly$\left(\frac{\log(n)}{\epsilon}\right)$ space. 
\end{theorem}

\begin{Frame}[\textbf{Algorithm \ref{alg:naive_approximation} : Na\"ive Approximation.}]
\label{alg:naive_approximation}
\textbf{Input:} Given a turnstile stream $\mathcal{P}$ with weighted unit intervals, where the weights are polynomially bounded, $\epsilon$ and $\delta >0$, Na\"ive Approximation outputs a $(9+\epsilon)$-approximation to $\beta$ with probability $1-\delta$.

\begin{enumerate}
	\item Randomly shift a grid $\Delta$ of side length $1$. Partition the cells into even and odd, denoted by $\mathcal{C}_e$ and $\mathcal{C}_o$. 
	\item Consider a partition of cells in $\mathcal{C}_e$ into $b = \textrm{poly}(\log(n))$ weight classes $\mathcal{W}_i = \{ c \in \mathcal{C}_e | (1+1/2)^i \leq m(c) < (1+1/2)^{i+1} \}$, where $m(c)$ is the maximum weight of an interval in $c$ (this is not an algorithmic step since we do not know this partition a priori). Create a substream for each weight class $\mathcal{W}_i$ denoted by $\mathcal{W}'_i$. 
    \item For each new interval $D(d_j,1, w_j)$, feed it to substream $\mathcal{W}'_i$ if $w_j \in [(1+1/2)^i, (1+1/2)^{i+1})$. For each substream $\mathcal{W}'_i$, maintain a $(1\pm\epsilon)$-approximate $\ell_0$-estimator (described below). 
	\item Let $t_i$ be the $\ell_0$ estimate corresponding to $\mathcal{W}'_i$. Let $X_e = \frac{2}{9(1+\epsilon)} \sum_{i\in [b]} (1+1/2)^{i+1} t_i$.
    \item Repeat Steps 2-6 for the odd cells $\mathcal{C}_o$ to obtain the corresponding estimator $X_o$. 
\end{enumerate}

\noindent\textbf{Output:} max$(X_e, X_o)$
\end{Frame}

\textbf{A na\"ive approximation.} We start by describing a simple approach (Algorithm \ref{alg:naive_approximation}) to obtain a $9$-approximation. 
The algorithm proceeds by imposing a grid of side length $1$ and shifts it by a random integer. This is a standard technique used in geometric algorithms. We then snap each interval to the cell containing the center of the interval and partition the real line into odd and even cells. This partitions the input space such that intervals landing in distinct odd (even) cells are pairwise independent. Let $\mathcal{C}_e$ be the set of all even cells and $\mathcal{C}_o$ be the set of all odd cells. 

By averaging, either $|\text{\textsf{OPT}}_{\mathcal{C}_e}|$ or $|\text{\textsf{OPT}}_{\mathcal{C}_o}|$ is at least $\frac{\textsf{OPT}}{2}$, where $|\textsf{OPT}|$ is the max weight independent set of intervals. We develop an estimator that gives a $(1+\epsilon)$-approximation to $|\text{\textsf{OPT}}_{\mathcal{C}_e}|$ as well as $|\text{\textsf{OPT}}_{\mathcal{C}_o}|$.  Therefore, taking the max of the two estimators, we obtain a $(2+\epsilon)$-approximation to $|\textsf{OPT}|$.  

Having reduced the problem to estimating $|\text{\textsf{OPT}}_{\mathcal{C}_e}|$, we observe that for each even cell only the max weight interval landing in the cell contributes to $\textsf{OPT}_{\mathcal{C}_e}$. Then, partitioning the cells in $\mathcal{C}_e$ into poly$(\log(n))$ geometrically increasing weight classes based on the max weight interval in each cell and approximately counting the number of cells in each weight class suffices to estimate $|\text{\textsf{OPT}}_{\mathcal{C}_e}|$ up to a $(1+\epsilon)$-factor. 

Given such a partition, we can approximate the number of cells in each weight class by running an $\ell_0$ norm estimator. Estimating the $\ell_0$ norm of a vector in \emph{turnstile streams} is a well studied problem and a result of Kane, Nelson and Woodruff \cite{DBLP:conf/pods/KaneNW10} obtains a $(1\pm\epsilon)$-approximation in poly$(\frac{\log(n)}{\epsilon})$ space. However, we do not know the partition of the cells into the weight classes a priori and this partition can vary drastically over the course of a stream given that intervals can be deleted. Therefore, the main technical challenge is to simulate this partition in \emph{turnstile streams}.   

As a first attempt, consider a partition of cells in $\mathcal{C}_e$ into $b = \textrm{poly}(\log(n))$ weight classes $\mathcal{W}_i = \{ c \in \mathcal{C}_e | (1+1/2)^i \leq m(c) < (1+1/2)^{i+1} \}$, where $m(c)$ is the maximum weight of an interval in $c$. Create a substream for each weight class $\mathcal{W}_i$ and feed an input interval into this substream if its weight lies in the range $[(1+1/2)^i, (1+1/2)^{i+1})$. Let $t_i$ be the corresponding $\ell_0$ estimate for this substream. Approximate the contribution of $\mathcal{W}_i$ by $(1 + 1/2)^{i+1}\cdot t_i$. Sum up the estimates for all $i \in [b]$ to obtain an estimate for $|\text{\textsf{OPT}}_{\mathcal{C}_e}|$. 

We note that there are two issues with our algorithm. First, we overestimate the weight of intervals in class $\mathcal{W}_i$ by a factor of $3/2$ and second, for a given cell we sum up the weights of all intervals landing in it, instead of taking the maximum weight for the cell. In the worst case, we approximate the true weight of a contributing interval, $(3/2)^{i+1}$, with $\sum^{i}_{i'=1}(3/2)^{i'+1} \leq 3((3/2)^{i+1}-1)$. Note, we again overestimate the weight, this time by a factor of $3$. Combined with the approximation for the $\ell_0$ norm, we obtain a weaker $(\frac{9}{2} + \epsilon)$-approximation to $|\text{\textsf{OPT}}_{\mathcal{C}_e}|$ in the desired space. From our discussion above, this implies a $(9+\epsilon)$-approximation to $|\textsf{OPT}|$.  
We also note that this attempt is not futile as we use the above algorithm as a subroutine subsequently.

\textbf{A refined attempt.}
Next, we describe an algorithm that estimates $|\text{\textsf{OPT}}_{\mathcal{C}_e}|$ up to a $(1+\epsilon)$-factor. Here, we use more sophisticated techniques to simulate a finer partition of the cells in $\mathcal{C}_e$ into geometrically increasing weight classes in turnstile streams.  
One key algorithmic tool we use here is a streaming algorithm for $k$-\textsf{Sparse Recovery}:  
given an input vector $x$ such that $x$ receives coordinate-wise updates in the turnstile streaming model and has at most $k$ non-zero entries at the end of the stream of updates,  
there exist data structures that exactly recover $x$ at the end of the stream. As mentioned in Berinde et al. \cite{BCIS09}, the $k$-tail guarantee is a sufficient condition for $k$-\textsf{Sparse Recovery}, since in a $k$-sparse vector, the elements of the tail are $0$. We note that the \textsf{Count-Sketch} Algorithm \cite{CCF02} has a $k$-tail guarantee in \emph{turnstile streams}.

This time around, we consider partitioning cells in $\mathcal{C}_e$ into $\text{poly}\left(\epsilon^{-1} \log(n)\right)$ weight classes, creating a substream for each one and computing the corresponding $\ell_0$ norm.  
We also assume we know $|\text{\textsf{OPT}}_{\mathcal{C}_e}|$ up to a constant (this can be simulated in \emph{turnstile streams}).
Formally, given $b=\text{poly}\left(\log(n),\epsilon^{-1}\right)$ weight classes, for all $i \in [b]$, let $\mathcal{W}_i$ denote the set of even cells with maximum weight sandwiched in the range $[(1+\epsilon)^{i}, (1+\epsilon)^{i+1})$. We then simulate sampling from the partition by subsampling cells in each $\mathcal{W}_i$ at the start of the stream, agnostic to the input. We do this at different sampling rates , i.e. for all $i \in[b]$, we subsample the cells in $\mathcal{W}_i$ with probability roughty $(1+\epsilon)^{i}/|\textsf{OPT}_{\mathcal{C}_e}|$.
%Formally, we subsample cells in $\mathcal{C}_e$ $b$ times, where $b = \text{poly}\left(\log(n),\epsilon^{-1}\right)$, with probability proportional to $\Theta( \frac{b (1+\epsilon)^i \log(n)}{\epsilon^3 |\text{\textsf{OPT}}_{\mathcal{C}_e}|})$, for all $i \in [b]$. 

This presents several issues, as we cannot subsample non-empty cells in turnstile streams a priori. Further, if a weight class has a small number of non-empty cells, we cannot recover accurate estimates for the contribution of this weight class to $|\text{\textsf{OPT}}_{\mathcal{C}_e}|$ at any level of the subsampling. 
To address the first issue, we agnostically sample cells from $\mathcal{C}_e$ according to a carefully chosen range of sampling rates and create a substream for each one. 
We then run a sparse recovery algorithm on the resulting substreams. At the right subsampling rate, we note that the resulting substream is sparse since we can filter out cells that belong to smaller weight classes. Further, we can ensure that the number of cells that survive from the relevant weight class (and larger classes) is small. Therefore, we recover all such cells using the sparse recovery algorithm. 

To address the second issue, we threshold the weight classes that we consider in the algorithm based on the relative fraction of non-empty cells in them. This threshold can be computed in the streaming algorithm using the $\ell_0$-norm estimates for each weight class. All the weight classes below the threshold together contribute at most an $\epsilon$-fraction of $|\text{\textsf{OPT}}_{\mathcal{C}_e}|$ and though we cannot achieve concentration for such weight classes, we show that we do not overestimate their contribution. Further, for all the weight classes above the threshold, we can show that sampling at the right rate can recover enough cells to achieve concentration.   

We complement the above algorithmic result with a matching lower bound, i.e., a $(2-\epsilon)$-approximation to \textsf{MIS}, for any $\epsilon > 0$, requires $\Omega(n)$ space. This follows from an easy application of the Augmented Indexing problem. We note that our result combined with the $3/2$-approximation by \cite{cabello2017interval} implies an unexpected separation between insertion-only and turnstile streams.

\subsection{Parametrized Algorithms for Arbitrary Length Intervals}

In light of the lower bound discussed above, we identify two sources contributing to the streaming hardness of \textsf{MIS} for arbitrary length intervals : the number of pair-wise intersections (max-degree) and the ratio of the longest to shortest interval (scale). We show that when either of these quantities is poly-logarithmically bounded, we can approximate \textsf{MIS} for arbitrary length intervals. 

Instead of assuming the max-degree or scale is bounded, we instead provide algorithms paramterized by these quantities. First, let the number of pair-wise intersections be
bounded by $\kappa_{\max}$. Then,

\begin{theorem}[Theorem \ref{thm:bounded_deg_ins_del}, informal.]
For $\epsilon >0$,  there exists a turnstile streaming algorithm that takes as input a set of unit-weight arbitrary-length intervals, with at most $\kappa_{\max}$ pair-wise intersections and with probability $99/100$, outputs a $(1+\epsilon)$-approximation to \textsf{MIS} in $\textrm{poly}(\log(n), \epsilon^{-1},\kappa_{\max})$ space. 
\end{theorem}

This result requires several new algorithmic ideas. Observe, placing a unit grid no longer suffices since the intervals now span different lengths. Therefore, we impose a nested grid on our input, where the grid size is geometrically increasing, and randomly shift it. Further, observe that the natural strategy that partition the interval into geometrically increasing \textit{length classes} and estimates each partition up to $1+\epsilon$ does not work since the intervals overlap. 

We therefore define the following object that uniquely determines intervals of a particular \textit{length class} contributing to the \textsf{MIS} :

\begin{definition}($r_i$-\textsf{Structure.})
\label{def:structure}
We define an $r_i$-\textsf{Structure} to be a subset of the \textsf{Nested Grid}, such that there exists an interval at the $i^{th}$ grid level, there exist no intervals in the grid at any level $i' > i$ and all the intervals in the grid at levels $i'< i$ intersect the interval at the $i^{th}$ level. 
\end{definition}

It is easy to see any interval that contributes to \textsf{MIS} corresponds to an $r_i$-\textsf{Structure} for some $i$. Therefore, it suffices to estimate the number of $r_i$-\textsf{Structures} for all $i$.  Following our approach for unit intervals, we again use $k$-\textsf{Sparse Recovery} as our main tool. At a high level,  we sub-sample poly$(\log(n), \epsilon^{-1})$ $r_i$-\textsf{Structures} from the set of all such structures at level $i$, and create a new substream for each $i$. We then run a $\kappa_{\textrm{max}}$-\textsf{Sparse Recovery} Algorithm on each substream.  We show that at the end of the stream, we obtain an estimate of the number of $r_i$-\textsf{Structures} at level $i$ that concentrates. Since the structures form a partition, our overall estimate is simply the sum of the estimates obtained for each $i$. 

The main algorithmic challenge here is to show that we can indeed detect and subsample the $r_i$-\textsf{structures}. These structures are defined in a way that takes into account how many intervals appear in the nested grid both above and below a given interval. Therefore, it is unclear how to track such updates as they constantly change over the stream. However, observe that since our space is parameterized by the max-degree, we can afford to store an $r_i$-\textsf{Structure} completely in memory. 

Given a randomly sampled cell from the $i$-th level of the nested grid, we assume this cell contributes an $r_i$-\textsf{structure}. We then run $\kappa_{\max}$-\textsf{Sparse Recovery} on this cell.  Our main insight is that at the end of the stream we can verify whether this cell indeed contributed an $r_i$-\textsf{structure} since we recover the nested intervals \textit{exactly}. The final remaining challenge is to ensure that our sub-sample contains a sufficient number of non-empty structures for each level and the resulting estimate concentrates. We describe these details in Section \ref{subsection:bounded_degree}.

Finally, we show that similar algorithmic ideas also result in a turnstile streaming algorithm,  if parametrize the input by the $W_{\max}$, the ratio of the largest to smallest interval : 

\begin{theorem}[Theorem \ref{thm:parametrized_space_ins_del}, informal.]
For $\epsilon >0$,  there exists a turnstile streaming algorithm that takes as input a set of unit-weight arbitrary-length intervals, with $W_{\max}$ being an upper bound on the ratio of the largest to smallest interval, and with probability $99/100$, outputs a $(2+\epsilon)$-approximation to \textsf{MIS} in $\textrm{poly}(\log(n), \epsilon^{-1}, W_{\max})$ space. 
\end{theorem}

\subsection{Unit-Radius Disks}

We generalize the \textsf{WMIS} turnstile streaming algorithm for unit length intervals to unit radius disks in $\mathbb{R}^2$. The approximation ratio for disks is closely related to the optimal circle packing constant. We leverage the hexagonal packing of circles in the 2-D plane to obtain the following result:

\begin{theorem}[Theorem \ref{thm:turnstile_geometric_disks}, informal]
There exists a turnstile streaming algorithm achieving a $\left(\frac{8\sqrt{3}}{\pi}+\epsilon\right)$-approximation to estimate $\textsf{WMIS}$ of unit disks with constant probability and in poly$\left(\frac{\log(n)}{\epsilon}\right)$ space. 
\end{theorem}

We note that a greedy algorithm for unweighted disks obtains a $5$-approximation to $\alpha$ \cite{erlebach2003maximum} and the space required is $O\left(\alpha\right)$. The greedy algorithm can be extended to obtain a $(5 + \epsilon)$-approximation in poly$\left(\frac{\log n }{\epsilon}\right)$ space using the sampling approach we presented in Section \ref{sec:turnstile_unit_interval}. However, beating the approximation ratio achieved by the greedy algorithm requires geometric insight.
Critically, we use the hexagonal packing of unit circles in a plane introduced by Lagrange \footnote{See \url{https://en.wikipedia.org/wiki/Circle_packing}}, which was shown to be optimal by Toth \cite{chang2010simple}. 

The hexagonal packing covers a $\frac{\pi}{\sqrt{12}}$ fraction of the area in two dimensions. We then partition the unit circles in the hexagonal packing into equivalence classes such that two circles in the same equivalence class are at least a unit distance apart. Formally, let $c_1, c_2$ be two unit circles in the hexagonal packing of the plane lying in the same equivalence class. Then, for all points $p_1 \in c_1$, $p_i \in c_2$, $\|p_1 - p_2 \|_2 \geq 1$. Therefore, if two input disks of unit radius have centers lying in distinct circles belong to the same equivalence class, the disks must be independent, as long as the disk are not centered on the boundary of the circles. 

Randomly shifting the underlying hexagonal packing ensures this happens with probability $1$. We then show that we can partition the hexagonal packing into four equivalence classes such that their union covers all the circles in the packing and disks lying in distinct circles of the same equivalence class are independent. 
%Further, the space required is proportional to the size of the optimal solution, $\alpha$. We then extend the sampling techniques developed unit-length intervals to convert this algorithm into one that runs in $\poly{\frac{\log(n)}{\epsilon}}$ space.

Algorithmically, we first impose a grid $\Delta$ of the hexagonal packing of circles with radius $1$ and shift it by a random integer. We discard all disks that do not have centers lying inside the grid $\Delta$.   
Given that a hexagonal packing covers a $\pi/\sqrt{12}$ fraction of the area, in expectation, we discard a $( 1 - \pi/\sqrt{12})$ fraction of $|\textsf{OPT}|$. We note that if we could accurately estimate the remaining \textsf{WMIS}, and scale the estimator by $\sqrt{12}/\pi$, we would obtain a $(\sqrt{12}/\pi)$-approximation to $|\textsf{OPT}|$. Let $||\textsf{OPT}|_{\textrm{hp}}|$ denote the remaining \textsf{WMIS}.  By Theorem \ref{athm:ins_del_unit_lowerbound} such an approximation requires $\Omega(n)$ space. 

\begin{figure*}
\centering
\textbf{Hexagonal Packing of Circles in the Plane}\par\medskip
%\vspace{-0.1in}
\includegraphics[width=0.5\textwidth]{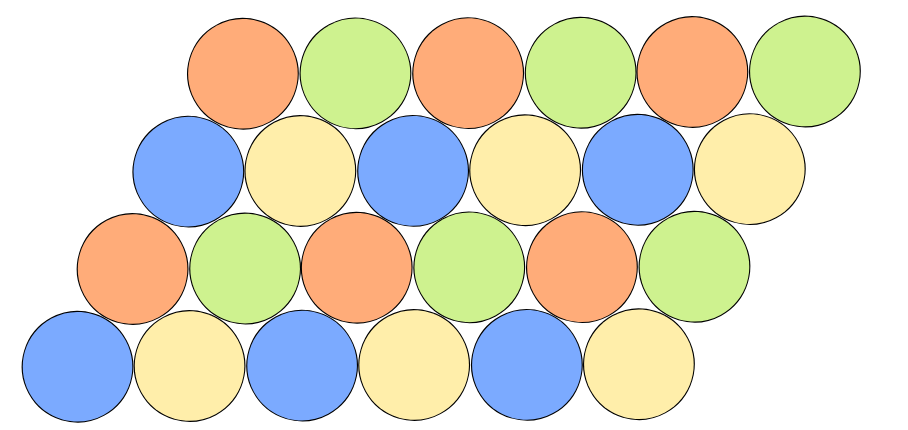}
\caption{We illustrate the hexagonal circle packing in the Euclidean Plane. Each color represents an equivalence class. Observe that input disks that are centered in distinct circles of the same equivalence class are independent, since the circles are at least $2$ units apart.} 
\label{fig:hexagonal_packing}
\end{figure*}

We then observe that the hexagonal circle packing grid can be partitioned into four equivalence classes. We use $\mathcal{C}_1, \mathcal{C}_2, \mathcal{C}_3$ and $\mathcal{C}_4$ to denote these equivalence classes. Since the equivalence classes form a partition of the hexagonal packing, at least one of them must contain a  $1/4$-fraction of the remaining maximum independent set. W.l.o.g, let $\mathcal{C}_1$ be the partition that contributes the most to $|\textsf{OPT}|$. Then, $|\textsf{OPT}_{\mathcal{C}_1}| \geq \frac{1}{4} |\textsf{OPT}_{\textrm{hp}}|$.
Therefore, we focus on designing an estimator for $\mathcal{C}_1$. We show a $(1+\epsilon)$-approximation to $\mathcal{C}_1$ in $\poly{\log(n), \epsilon^{-1}}$ space generalizing the algorithmic ideas we introduced for Theorem \ref{thm:weighted_ins_del_unit}. This implies an overall $\left(\frac{4 \sqrt{12}}{\pi} + \epsilon\right) = \left(\frac{8 \sqrt{3}}{\pi} + \epsilon \right)$ approximation for $|\textsf{OPT}|$.

\newpage
\bibliographystyle{alpha}
\bibliography{main_arxiv}

\appendix
\newpage
\noindent{\Huge\textbf{Appendix} }
\section{Weighted Interval Selection for Unit Intervals}
\label{sec:turnstile_unit_interval}
In this section, we present an algorithm to approximate the weight of the maximum independent set, $\beta$, for unit-length intervals in turnstile streams. Interestingly, we note that estimating $\beta$ has the same complexity as approximating $\alpha$ for unit-length intervals. That is, we obtain a $(2 + \epsilon)$-approximation to $\beta$ in the turnstile model, which immediately implies $(2+\epsilon)$-approximation for $\alpha$, where the weights are identical. We complement this result with a lower bound that shows any $(2-\epsilon)$-approximation to $\alpha$ requires $\Omega(n)$ space. The main algorithmic guarantee we achieve is as follows:
 
\begin{theorem}
\label{thm:weighted_ins_del_unit}
Let $\mathcal{P}$ be a turnstile stream of weighted unit intervals such that the weights are polynomially bounded in $n$ and let $\epsilon \in (0, 1/2)$. There exists an algorithm that outputs an estimator $Y$ such that with probability at least $9/10$ the following guarantees hold:
\begin{enumerate}
    \item $\frac{\beta}{2(1+\epsilon)} \leq Y \leq \beta$.
    \item The total space used is $\text{poly}\left(\frac{\log(n)}{\epsilon}\right)$. 
\end{enumerate}
\end{theorem}

We first impose a grid $\Delta$ of side length $1$ and shift it by a random integer. We then snap each interval to the cell containing the center of the interval and partition the real line into odd and even cells. Let $\mathcal{C}_e$ be the set of all even cells and $\mathcal{C}_o$ be the set of all odd cells. By averaging, either $|\text{\texttt{OPT}}_{\mathcal{C}_e}|$ or $|\text{\texttt{OPT}}_{\mathcal{C}_o}|$ is at least $\frac{\beta}{2}$. We describe an estimator that gives a $(1+\epsilon)$-approximation to $|\text{\texttt{OPT}}_{\mathcal{C}_e}|$ and $|\text{\texttt{OPT}}_{\mathcal{C}_o}|$. W.l.o.g let $|\text{\texttt{OPT}}_{\mathcal{C}_e}| \geq |\text{\texttt{OPT}}_{\mathcal{C}_o}|$.  
%We develop an estimator that gives a $(1+\epsilon)$-approximation to $|\text{\texttt{OPT}}_{\mathcal{C}_e}|$ as well as $|\text{\texttt{OPT}}_{\mathcal{C}_o}|$. 
Therefore, taking the max of the two estimators, we obtain a $(2+\epsilon)$-approximation to $\beta$.

Having reduced the problem to estimating $|\text{\texttt{OPT}}_{\mathcal{C}_e}|$, we observe that each even cell has at most $1$ interval, namely the max weight interval landing in the cell, contributing to $\texttt{OPT}_{\mathcal{C}_e}$. Then, partitioning the cells in $\mathcal{C}_e$ into poly$(\log(n))$ weight classes based on the max weight interval in each cell and approximately counting the number of cells in each weight class suffices to estimate $|\text{\texttt{OPT}}_{\mathcal{C}_e}|$ up to a $(1+\epsilon)$-factor. Given such a partition, we can create a substream for each weight class in the partition and compute the $\ell_0$ norm of each substream. However, we do not know the partition of the cells into the weight classes a priori and this partition can vary drastically over the course of stream given that intervals can be deleted. The main technical challenge is to simulate this partition. A key tool we use is to estimate the $\ell_0$ norm of a vector in turnstile streams is a well studied problem and we use a result of Kane, Nelson and Woodruff \cite{DBLP:conf/pods/KaneNW10} to obtain a $(1\pm\epsilon)$-approximation in poly$(\frac{\log(n)}{\epsilon})$ space. 

\begin{theorem}{($\ell_0$-Norm Estimation \cite{DBLP:conf/pods/KaneNW10}.)}
\label{thm:ell_0}
In the turnstile model, there is an algorithm for $(1\pm \epsilon)$-approximating the $\ell_0$-norm (number of non-zero coordinates) of a vector using space $\text{poly}\left(\frac{\log(n)}{\epsilon}\right)$ with success probability $2/3$. 
\end{theorem}

We begin by describing a simple algorithm which obtains a weaker $(9/2+\epsilon)$-approximation to $|\text{\texttt{OPT}}_{\mathcal{C}_e}|$ and in turn a $(9+\epsilon)$-approximation to $\beta$. Formally, consider a partition of cells in $\mathcal{C}_e$ into $b = \textrm{poly}(\log(n))$ weight classes $\mathcal{W}_i = \{ c \in \mathcal{C}_e | (1+1/2)^i \leq m(c) < (1+1/2)^{i+1} \}$, where $m(c)$ is the maximum weight of an interval in $c$. Create a substream for each weight class $\mathcal{W}_i$, denoted by $\mathcal{W}'_i$, and feed an input interval into this substream if it's weight lies in the range $[(1+1/2)^i, (1+1/2)^{i+1})$. Let $t_i$ be the corresponding $\ell_0$ estimate for substream $\mathcal{W}'_i$. Then, we can approximate the contribution of $\mathcal{W}_i$ by $(1 + 1/2)^{i+1}\cdot t_i$. Summing over the $b$ weight classes gives an estimate for $|\text{\texttt{OPT}}_{\mathcal{C}_e}|$.

Given access to an algorithm for estimating the $\ell_0$-norm, the Na\"ive Approximation Algorithm (\ref{alg:naive_approximation}) satisfies the following guarantee:

\begin{lemma}
\label{lem:NaiveApprox}
The Na\"ive Approximation Algorithm (\ref{alg:naive_approximation}) outputs an estimate $X$ such that with probability $99/100$, $\frac{\beta}{9(1+\epsilon)} \leq X \leq \beta$ and runs in space poly$\left(\frac{\log(n}{\epsilon}\right)$. 
\end{lemma}

\begin{proof}
We observe that for each non-empty cell $c \in \mathcal{C}_e$, there is exactly $1$ interval that can contribute to $|\text{\texttt{OPT}}_{\mathcal{C}_e}|$ since each cell of the grid has side length $1$ and all intervals falling in a given cell pairwise intersect. This contributing interval lies in some weight class $\mathcal{W}_i$ and our estimator approximates its weight as $(1+1/2)^{i+1}$. Here, the weights of the intervals are sandwiched between $(1+1/2)^i$ and $(1+1/2)^{i+1}$. Therefore, we overestimate the weight by a factor of at most $3/2$. 

Further, instead of taking the maximum over each cell $c$, we may have inserted intervals that lie in $c$ into all substreams $\mathcal{W}'_i$. Therefore, we take the sum of our geometrically increasing weight classes over that cell. In the worst case, we approximate the true weight of a contributing interval, $(3/2)^{i+1}$, with $\sum^{i}_{i'=1}(3/2)^{i'+1} = 3((3/2)^{i+1}-1)$. Note, we again overestimate the weight, this time by a factor of $3$. 

Finally, Theorem \ref{thm:ell_0} overestimates the $\ell_0$-norm of $\mathcal{W}_i$ by at most $1+\epsilon$ with probability at least $2/3$. We boost this probability by running $O(\log(n))$ estimators and taking the median. Union bounding over all $i \in [b]$, we simultaneously overestimate the $\ell_0$-norm of all $\mathcal{W}_i$ by at most $1+\epsilon$ with probability at least $99/100$. Therefore, the overall estimator is a $(9/2+\epsilon)$-approximation to $|\text{\texttt{OPT}}_{\mathcal{C}_e}|$. Rescaling our estimator by the above constant underestimates $|\text{\texttt{OPT}}_{\mathcal{C}_e}|$. Finally, $|\text{\texttt{OPT}}_{\mathcal{C}_e}|\geq \beta/2$ and $\frac{\beta}{(9+\epsilon)} \leq X \leq \beta$. 

Since our weights are polynomially bounded, we create poly$\left(\log_{1+\epsilon}(n)\right)$ substreams and run an $\ell_0$ estimator from Theorem \ref{thm:ell_0} on each substream. Therefore, the total space used by Algorithm \ref{alg:naive_approximation} is poly$\left(\frac{\log(n}{\epsilon}\right)$.
\end{proof}

We can thus assume we know $\beta$ and $|\text{\texttt{OPT}}_{\mathcal{C}_e}|$ up to a constant by initially making $O\left(\log(n)\right)$ guesses and running the Na\"ive Approximation Algorithm for each guess in parallel. At the end of the stream, we know the correct guess up to a constant factor, and thus can output the estimator corresponding to that branch of computation.
A key tool we use in this algorithm is $k$-\texttt{Sparse Recovery}. As mentioned in Berinde et al. \cite{BCIS09}, the $k$-tail guarantee is a sufficient condition for $k$-\texttt{Sparse Recovery}, since in a $k$-sparse vector, the elements of the tail are $0$. We note that the \texttt{Count-Sketch} Algorithm \cite{CCF02} has a $k$-tail guarantee in turnstile streams.

\begin{definition}($k$-\texttt{Sparse Recovery}.)
\label{def:sparse_rec}
Let $x$ be the input vector such that $x$ is updated coordinate-wise in the turnstile streaming model. Then, $x$ is $k$-sparse if $x$ has at most $k$ non-zero entries at the end of the stream of updates. 
Given that $x$ is $k$-sparse, a data structure that exactly recovers $x$ at the end of the stream is referred to as a $k$-\texttt{Sparse Recovery} data structure. 
\end{definition}

%The key ingredient in our algorithm is computing the $\ell_1$-heavy hitters in a turnstile stream using the Count-Min Sketch, introduced by Cormode and Muthukrishnan \cite{DBLP:conf/latin/CormodeM04} and obtain the following guarantee:

%\begin{theorem}($\ell_1$ Heavy Hitters.)
%\label{athm:cm_sketch}
%Given $\epsilon, \delta, \phi >0$ and a vector $a$, there exists an algorithm that outputs every item $a_i$ such that $|a_i| \geq (\phi+\epsilon) |a|_1$, and no item $a_i$ such that $|a_i| < \phi |a|_1$ with probability at least $\delta$, and uses space $O\left(\frac{\log(n)}{\epsilon} \log( \frac{2\log(n)}{\delta \phi})\right)$.
%\end{theorem}

Intuitively, we again simulate partitioning cells in $\mathcal{C}_e$ into $\text{poly}\left(\frac{\log(n)}{\epsilon}\right)$ weight classes according to the maximum weight occurring in each cell. Since we do not know this partition a priori, we initially create $b =O\left(\frac{\log(n)}{\epsilon}\right)$ substreams, one for each weight class and run the $\ell_0$-estimator on each one. We then make $O\left(\frac{\log(n)}{\epsilon}\right)$ guesses for $|\text{\texttt{OPT}}_{\mathcal{C}_e}|$ and run the rest of the algorithm for each branch in parallel. Additionally, we run the Na\"ive Approximation Algorithm to compute the right value of $|\text{\texttt{OPT}}_{\mathcal{C}_e}|$ up to a constant factor, which runs in space $\text{poly}\left(\frac{\log(n)}{\epsilon}\right)$. Then, we create $b = \text{poly}\left(\frac{\log(n)}{\epsilon}\right)$ substreams by agnostically sampling cells with probability $p_i = \Theta\left( \frac{b (1+\epsilon)^i \log(n)}{\epsilon^3 X}\right)$, where $X$ is the right guess for $|\text{\texttt{OPT}}_{\mathcal{C}_e}|$. Sampling at this rate preserves a sufficient number of cells from weight class $\mathcal{W}_i$. We then run a sparse recovery algorithm on the resulting substreams. 

We note that the resulting substreams are sparse. To see this, note we can filter out cells that belong weight classes $\mathcal{W}_{i'}$ for $i' < i$ by simply checking if the maximum interval seen so far lies in weight classes $\mathcal{W}_{i}$ and higher. Further, sampling with probability proportional to $\Theta\left( \frac{b (1+\epsilon)^i \log(n)}{\epsilon^3 |\text{\texttt{OPT}}_{\mathcal{C}_e}|}\right)$ ensures that the number of cells that survive from weight classes $\mathcal{W}_{i}$ and above are small. Therefore, we recover all such cells using the sparse recovery algorithm. Note, we limit the algorithm to considering weight classes that have a non-trivial contribution to $\text{\texttt{OPT}}_{\mathcal{C}_e}$. 

Using the $\ell_0$ norm estimates computed above, we can determine the number of non-empty cells in each of the weight classes. Thus, we create a threshold for weight classes that contribute, such that all the weight classes below the threshold together contribute at most an $\epsilon$-fraction of $|\text{\texttt{OPT}}_{\mathcal{C}_e}|$ and we can set their corresponding estimators to $0$. Further, for all the weight classes above the threshold, we can show that sampling at the right rate leads to recovering enough cells to achieve concentration in estimating their contribution. 

%Subsequently, we show that running a $\text{poly}\left(\frac{\log(n)}{\epsilon}\right)$-sparse recovery algorithm on each substream recovers the cells that correspond the the weight class $\mathcal{W}_i$ and filters out the cells that belong to other weight classes. 

%We then simulate the partition by subsampling even cells at different sampling rates. Formally, we subsample cells in $\mathcal{C}_e$ $b = \text{poly}\left(\frac{\log(n)}{\epsilon}\right)$ times, with probability proportional to $\Theta( \frac{b (1+\epsilon)^i \log(n)}{\epsilon^3 |\text{\texttt{OPT}}_{\mathcal{C}_e}|})$, for all $i \in [b]$. This presents several issues, as we cannot subsample non-empty cells in turnstile streams a priori.  Further, if a weight class has a small number on non-empty cells, we cannot recover accurate estimates for the contribution of this weight class to $|\text{\texttt{OPT}}_{\mathcal{C}_e}|$ at any level of the subsampling. To address the first issue, 

%  We do not know this partition a priori, therefore we subsample cells from the input space for a range of probabilities. 
%We then sample a small number of non-empty cells for each weight class and find the heavy hitter cells in each sample set. To sample non-empty cells with a probability $p$ in turnstile streams, we create a substream by sampling every cell in the input space with probability $p$ and running a heavy hitters algorithm on this stream. Note, the empty cells are ignored by the heavy hitters algorithm. We show that an estimator constructed from the heavy hitters suffices.  

Next, we show that the total space used by Algorithm \ref{alg:weighted_unit_interval_sampling} is $\text{poly}\left(\frac{\log(n)}{\epsilon}\right)$. We initially create $b =O\left(\frac{\log(n)}{\epsilon}\right)$ substreams, one for each weight class and run an $\ell_0$-estimator on each one. Recall, this requires $\text{poly}\left(\frac{\log(n)}{\epsilon}\right)$. We then make $O\left(\frac{\log(n)}{\epsilon}\right)$ guesses for $|\text{\texttt{OPT}}_{\mathcal{C}_e}|$ and run the rest of the algorithm for each branch in parallel. Additionally, we run Algorithm \ref{alg:naive_approximation} to compute the right value of $|\text{\texttt{OPT}}_{\mathcal{C}_e}|$ up to a constant factor, which runs in space $\text{poly}\left(\frac{\log(n)}{\epsilon}\right)$. Then, we create $b$ substreams by sampling cells with probability $p_i = \Theta\left( \frac{b (1+\epsilon)^i \log(n)}{\epsilon^3 X}\right)$, for $i \in [b]$. Subsequently, we run a $\text{poly}\left(\frac{\log(n)}{\epsilon}\right)$-sparse recovery algorithm on each one. Note, if each sample is not too large, this can be done in $\text{poly}\left(\frac{\log(n)}{\epsilon}\right)$ space. Therefore, it remains to show that each sample $\mathcal{S}_i$ is small. 

\begin{lemma}
\label{lem:weighted_ins_del_unit_space}
Given a turnstile stream $\mathcal{P}$, with probability at least $99/100$, the Weighted Unit Interval Turnstile Sampling procedure (Algorithm \ref{alg:weighted_unit_interval_sampling}) samples $\text{poly}\left(\frac{\log(n)}{\epsilon}\right)$ cells from the grid $\Delta$.  
\end{lemma}
\begin{proof}
For $i \in [b]$, let $\mathcal{S}_i$ be a substream of cells in $\mathcal{C}_e$, sampled with probability $p_i$ and having an interval with weight at least $(1+\epsilon)^i$ since we filter out all cells with smaller weight. Then, by an averaging argument, the total number of cells with an interval of weight at least $(1+\epsilon)^{i}$ is at most $\frac{\beta}{(1+\epsilon)^i}$. Sampling with probability $p_i = \Theta\left( \frac{b (1+\epsilon)^i \log(n)}{\epsilon^3 X}\right)$, the expected number of cells from $\mathcal{W}_i$ that survive in $\mathcal{S}_i$ is at most $p_i \frac{\beta}{(1+\epsilon)^i} = \text{poly}\left(\frac{\log(n)}{\epsilon}\right)$ in expectation. Next, we show that they are never much larger than their expectation. Let $X_c$ be the indicator random variable for cell $c \in \mathcal{W}_i$ to be sampled in $\mathcal{S}_i$ and let $\mu$ be the expected number of cells in $\mathcal{S}_i$. By Chernoff bounds, 
$$
\Pr\left[\sum_c X_c \geq (1+\epsilon)\mu\right] \leq \textrm{exp}\left(-\frac{2\epsilon^2 \textrm{poly}(\log(n))}{\textrm{poly}(\epsilon)}\right) \leq 1/n^k
$$
for some large constant $k$. A similar argument holds for the number of cells from weight class $\mathcal{W}_{i'}$, for $i' > i$, surviving in substream $\mathcal{S}_i$. Note, for all $i' < i$, we never include such a cell from weight class $\mathcal{W}_{i'}$ in our sample $\mathcal{S}_i$, since the filtering step rejects all cells that do not contain an interval of weight at least $(1+\epsilon)^i$. Union bounding over the events that cells $c \in \mathcal{W}_{i'}$ get sampled in $\mathcal{S}_{i}$, for $i'\geq i$, the cardinality of $\mathcal{S}_i$ is at most $\text{poly}\left(\frac{\log(n)}{\epsilon}\right)$ with probability at least $1-1/n^{k'}$ for an appropriate constant $k'$. Since we create $b$ such substreams for $\mathcal{C}_e$, we can union bound over such events in each of them and thus $\bigcup_{i\in [b]} |\mathcal{S}_i|$ is at most $\text{poly}\left(\frac{\log(n)}{\epsilon}\right)$ with probability at least $99/100$. Since $|\mathcal{C}_e|$ is $|\Delta|/2$, the same result holds for the total cells sampled from $\Delta$. Therefore, the overall space used by Algorithm \ref{alg:weighted_unit_interval_sampling} is $\text{poly}\left(\frac{\log(n)}{\epsilon}\right)$.  
\end{proof}

Next, we show that the estimate returned by our sampling procedure is indeed a $(2 + \epsilon)$-approximation. We observe that the union of the $\mathcal{W}_i$'s form a partition of $\mathcal{C}_e$. Therefore, it suffices to show that we obtain a $(1+\epsilon)$-approximation to the \texttt{WIS} for each $\mathcal{W}_i$ with good probability. Let $c$ denote a cell in $\mathcal{W}_i$ and $\text{\texttt{OPT}}_c$ denote the \texttt{WIS} in cell $c$. We create a substream for each weight class $\mathcal{W}_i$ denoted by $\mathcal{W}'_i$ and let $X_{\mathcal{W}'_i}$ be the corresponding estimate returned by the $\ell_0$ norm of $\mathcal{W}'_i$. Let $X_{\mathcal{W}'} = \sum_{i \in [b]} X_{\mathcal{W}'_i}$ denote the sum of the estimates across the $b$ substreams. 

We say that weight class $\mathcal{W}_i$ contributes if $X_{\mathcal{W}'_i} \geq \frac{\epsilon X_{\mathcal{W}'} }{(1+\epsilon)^{i+1}b}$. 
Note, if we discard all the weight classes that do not contribute we lose at most an $\epsilon$-fraction of $\beta$ (as shown below). Therefore, setting the estimators corresponding to classes that do not contribute to $0$ suffices. The main technical hurdle remaining is to show that if a weight class contributes we can accurately estimate $|\texttt{OPT}_{\mathcal{W}_i}|$.  

\begin{lemma}
\label{lem:weighted_ins_del_unit_approx}
Let $Y_e = \sum_i Y_i$ be the estimator returned by Algorithm \ref{alg:weighted_unit_interval_sampling} for the set $\mathcal{C}_e$. Then, $Y_e =(1\pm\epsilon)|\text{\texttt{OPT}}_{\mathcal{C}_e}|$ with probability at least $99/100$.
\end{lemma}

\begin{proof}
We first consider the case when $\mathcal{W}_i$ contributes, i.e., $X_{\mathcal{W}'_i} \geq \frac{\epsilon X_{\mathcal{W}'} }{(1+\epsilon)^{i+1}b}$. Note, $X_{\mathcal{W}'} = \sum_{i \in [b]} X_{\mathcal{W}'_i}$ is a $(1\pm \epsilon)$-approximation to the number of non-empty cells in $\mathcal{W}$ with probability at least $1 - n^{-k}$, where $\mathcal{W} = \bigcup_{i \in [b]} \mathcal{W}_i$, since the $\ell_0$-estimator is a $(1\pm \epsilon)$-approximation to the number of non-empty cells in $\mathcal{W}_i$ simultaneously for all $i$ with high probability and the $\mathcal{W}_i$'s are disjoint. Recall, $X$ is the correct guess for $|\text{\texttt{OPT}}_{\mathcal{C}_e}|$. Therefore, 
$$
(1+\epsilon)^i X_{\mathcal{W}'_i} = \Omega\left( \frac{\epsilon  X}{(1+\epsilon)b}\right) = \Omega\left(\frac{\epsilon X}{(1+\epsilon)b}\right)
$$
Then, sampling at a rate $p_i = \Theta( \frac{b (1+\epsilon)^i \log(n)}{\epsilon^3 X})$ implies at least $\Omega\left(\frac{\epsilon X}{(1+\epsilon)^{i+1} b}\right)\cdot\Theta( \frac{b (1+\epsilon)^i \log(n)}{\epsilon^3 X})  = \Omega\left(\frac{\log(n)}{(1+\epsilon)\epsilon^2}\right)$ cells from $\mathcal{W}_i$ survive in expectation. Let $X_c$ denote an indicator random variable for cell $c \in \mathcal{W}_i$ being in substream $\mathcal{S}_i$. Then, by a Chernoff bound,
$$
\Pr\left[ \sum_{c\in \mathcal{W}_i} X_c  \leq (1-\epsilon)\left(\frac{\log(n^{-c})}{2\epsilon^2}\right)\right] \leq \textrm{exp}\left(\frac{-2\epsilon^2 \log(n^{-c})}{2\epsilon^2}\right) \leq n^{-c}
$$
for some constant $c$. Union bounding over all the random events similar to the one above for $i\in [b]$, simultaneously for all $i$, the number of cells from $\mathcal{W}_i$ in $\mathcal{S}_i$ is at least $\Omega\left(\frac{\log(n)}{\epsilon^2}\right)$ with probability at least $1-1/n^k$ for some constant $k$. Note, for $i'< i$, no cell $c \in \mathcal{W}_{i'}$ exists in $\mathcal{S}_i$ since the filter step removes all cells $c$ such that $m(c) < (1+\epsilon)^i$.  

Next, consider a weight class $\mathcal{W}_{i'}$ for $i' > i$ such that it contributes. We upper bound the number of cells from $\mathcal{W}_{i'}$ that survive in substream $\mathcal{S}_i$. Note, weight class $\mathcal{W}_{i'}$ contains at most $\frac{\beta}{(1+\epsilon)^{i+1}}$ non empty cells for $i'>i$. In expectation, at most $\frac{\beta}{(1+\epsilon)^{i+1}}\cdot p_i =  O\left(b\frac{\log(n)}{(1+\epsilon)\epsilon^3}\right)$ cells from $\mathcal{W}_{i'}$ survive in sample $\mathcal{S}_i$, for $i'>i$. By a Chernoff bound, similar to the one above, simultaneously for all $i'>i$, at most $O\left(b\frac{\log(n)}{(1+\epsilon)\epsilon^3}\right)$ cells from $\mathcal{W}_{i'}$ survive, with probability at least $1-1/n^{k'}$.
%Further, in expectation, the fraction of $|\text{\texttt{OPT}}_{\mathcal{C}_e}|$ in the union of our samples $\mathcal{S}_i$ is $p_i |\text{\texttt{OPT}}_{\mathcal{C}_e}| = \Theta(\frac{(1+\epsilon)^i b \log(n)}{\epsilon^3})$. 

Now, we observe that the total number of cells that survive the sampling process in substream $\mathcal{S}_i$ is $\text{poly}\left(\frac{\log(n)}{\epsilon}\right)$ and therefore, they can be recovered exactly by the $\text{poly}\left(\frac{\log(n)}{\epsilon}\right)$-sparse recovery algorithm. Let the resulting set be denoted by $\mathcal{S}'_i$. We can also compute the number of cells that belong to weight class $\mathcal{W}_i$ that are recovered in the set $\mathcal{S}'_i$ and we denote this by $|\mathcal{S}'_{i | \mathcal{W}_i}|$. Recall, the corresponding estimator is $Y_i = \sum_{c \in \mathcal{S}'_{i| \mathcal{W}_i}}\frac{X_{\mathcal{W}'_i} Z_c}{\left|\mathcal{S}'_{i| \mathcal{W}_i}\right|}$, where $Z_c = \frac{(1+\epsilon)^{i+1}}{p_i}$ if $c \in \mathcal{S}'_{i| \mathcal{W}_i}$ and $0$ otherwise. We first show we obtain a good estimator for $|\texttt{OPT}_{\mathcal{W}_i}|$ in expectation. 

%for $i' \geq i $, we observe that any cell $c \in \mathcal{W}_{i'}$ which survives the sampling process has weight at least $(1+\epsilon)^i$, and therefore it is an $\ell_1$-heavy hitter with parameter $\phi_i$. By Theorem \ref{athm:cm_sketch}, we can recover all such cells $c$ with high probability. Additionally, for $i' < i$, we get no cell $c \in \mathcal{W}_{i'}$ since they cannot be heavy hitters and are removed from the sample. 
%Recall, $\Omega(\frac{\log(n)}{\epsilon^2})$ cells from $\mathcal{W}_i$ are sampled in $\mathcal{S}_i$, which is sufficient to show the concentration of $Y_i$ around $|\texttt{OPT}_{\mathcal{W}_i}|$. For $c \in \mathcal{W}_i$, let $Y^c_i$ be a random variable that is $m(c)/p_i$ if $c \in \mathcal{S}_i$ and $0$ otherwise and let $Y_i = \sum_{c\in \mathcal{S}_i} Y^c_i = \sum_{c \in \mathcal{W}_i} Y^c_i$. Observe, 

\begin{equation*}
\begin{split}
\E\left[ Y_i \right]  = \E \left[ \sum_{c \in \mathcal{S}'_{i| \mathcal{W}_i}}\frac{X_{\mathcal{W}'_i} Z_c}{\left|\mathcal{S}'_{i| \mathcal{W}_i}\right|} \right] &= (1+\epsilon)^{i+1} X_{\mathcal{W}'_i} \\
%& = (1\pm \epsilon) (1+\epsilon)^{i+1} |\mathcal{W}'_i| = (1\pm \epsilon)(1+\epsilon) |\text{\texttt{OPT}}_{\mathcal{W}_i}| \\ 
& = (1\pm 4\epsilon) |\text{\texttt{OPT}}_{\mathcal{W}_i}|
\end{split}
\end{equation*}

Since we know that $|\mathcal{S}'_{i | \mathcal{W}_i}| = $ $\Omega(\frac{\log(n)}{(1+\epsilon)\epsilon^2})$, we show that our estimator concentrates. Note, $\E\left[ Y_i \right] = (1+\epsilon)^{i+1} X_{\mathcal{W}'_i} = \Omega(\frac{\epsilon X}{\log(n)})$. Further, $0 \leq Z_c \leq \frac{(1+\epsilon)^{i+1}}{p_i} = O\left( \frac{(1+\epsilon)^{i+1} \epsilon^3 X}{b (1+\epsilon)^i \log(n)} \right)$. 
Therefore, $$ \mathlarger{\sum_{c \in \mathcal{S}'_{i| \mathcal{W}_i}}} \left(\frac{X_{\mathcal{W}'_i} Z_c}{\left|\mathcal{S}'_{i| \mathcal{W}_i}\right|} \right)^2 = O\left( \left(\frac{(1+\epsilon)^{i+1} \epsilon^3 X \cdot X_{\mathcal{W'}_i}}{b (1+\epsilon)^i \log(n)}\right)^2 \frac{1}{\left|\mathcal{S}'_{i| \mathcal{W}_i}\right|} \right)$$ and $X_{\mathcal{W}'_i} = \text{poly}\left(\frac{\log(n)}{\epsilon}\right)$ 
By a Hoeffding bound, 

\begin{equation*}
\begin{split}
\Pr\Big[ |Y_i - E\left[Y_i\right]|  \geq \epsilon E\left[ Y_i \right]\Big] & \leq 2 \textrm{exp}\left(\frac{-2\epsilon^2 E\left[ Y_i \right]^2}{\mathlarger{\sum_{c \in \mathcal{S}'_{i| \mathcal{W}_i}}} \left(\frac{X_{\mathcal{W}'_i} Z_c}{\left|\mathcal{S}'_{i| \mathcal{W}_i}\right|} \right)^2 }\right ) \\
& \leq 2 \textrm{exp}\left(\frac{\Omega(\log(n))}{1+\epsilon}\right) \leq 1/n^k
\end{split}
\end{equation*}
for some constant $k$. Therefore, union bounding over all $i$, $Y_i$ is a $(1 \pm \epsilon)^2$-approximation to $|\text{\texttt{OPT}}_{\mathcal{W}_i}|$ with probability at least $1 - 1/n$. Therefore, if $\mathcal{W}_i$ contributes we obtain a $(1 \pm \epsilon)$-approximation to $|\text{\texttt{OPT}}_{\mathcal{W}_i}|$.   

In the case where $\mathcal{W}_i$ does not contribute, we set the corresponding estimator to $0$. Note, $X_{\mathcal{W}'_i} < \frac{\epsilon X_{\mathcal{W}'} }{(1+\epsilon)^{i+1}b} = \frac{\epsilon (1\pm \epsilon) |\texttt{OPT}_{\mathcal{W}}|}{b} = O(\frac{\epsilon \beta}{b})$. Note, since there are at most $b$ weight classes, discarding all weight classes that do not contribute discards at most $O(\epsilon \beta)$. We therefore lose at most an $\epsilon$-fraction of $\beta$ by setting the $Y_i$ corresponding to non-contributing weight classes to $0$. 

%by a Chernoff bound, the number of samples we get is smaller than $O\left( \frac{\epsilon W}{b(1+\epsilon)^i}\right) $ with high probability. Therefore, our estimate $Y_i < \Theta\left(\frac{\log(n)}{\epsilon^3}\right)$ and can be set to $0$. Note, in this case we do not obtain a concentration, but we also do not overcount. The proof of Theorem \ref{athm:weighted_ins_del_unit} immediately follows from union bounding the two cases above and Lemma \ref{alem:weighted_ins_del_unit_space}.
\end{proof}
Combining Lemmas \ref{lem:NaiveApprox} and  \ref{lem:weighted_ins_del_unit_approx} finishes the proof for Theorem \ref{thm:weighted_ins_del_unit}.

\begin{Frame}[\textbf{Algorithm \ref{alg:weighted_unit_interval_sampling} : Weighted Unit Interval Turnstile Sampling.}]
\label{alg:weighted_unit_interval_sampling}
\textbf{Input:} Given a turnstile stream $\mathcal{P}$ with weighted unit intervals, where the weights are polynomially bounded, $\epsilon$ and $\delta >0$, the sampling procedure outputs a $(2 + \epsilon)$-approximation to $\beta$.

\begin{enumerate}
	\item Randomly shift a grid $\Delta$ of side length $1$. Partition the cells into $\mathcal{C}_e$ and $\mathcal{C}_o$.
    \item For cells in $\mathcal{C}_e$, snap each interval in the input to a cell $c$ that contains its center. Consider a partitioning of the cells in $\mathcal{C}_e$ into $b = \text{poly}\left(\frac{\log(n)}{\epsilon}\right)$ weight classes $\mathcal{W}_i = \{ c \in \mathcal{C}_e | (1+\epsilon)^i \leq m(c) \leq (1+\epsilon)^{i+1}$ \}, where $m(c)$ is the maximum weight of an interval in $c$ (we do not know this partition a priori.) Create a substream for each weight class $\mathcal{W}_i$ denoted by $\mathcal{W}'_i$. 
    \item Feed interval $D(d_j,1, w_j)$ along substream $\mathcal{W}'_i$ such that $w_j \in [(1+\epsilon)^i, (1+\epsilon)^{i+1})$. Maintain a $(1\pm\epsilon)$-approximate $\ell_0$-estimator for each substream. Let $|\mathcal{W}'_i|$ denote the number of non-empty cells in substream $\mathcal{W}'_i$ and $X_{\mathcal{W}'_i}$ be the corresponding estimate returned by the $\ell_0$-estimator. 
	\item Create $O(\log(n))$ substreams, one for each guess of $|\text{\texttt{OPT}}_{\mathcal{C}_e}|$. Let $X$ be the guess for the current branch of the computation. In parallel, run Algorithm \ref{alg:naive_approximation} estimates $|\text{\texttt{OPT}}_{\mathcal{C}_e}|$ up to a constant factor. Therefore, at the end of the stream, we know a constant factor approximation to the correct value of $|\text{\texttt{OPT}}_{\mathcal{C}_e}|$ and use the estimator from the corresponding branch of the computation.    
	\item In parallel, for $i \in [b]$, create substream $\mathcal{S}_i$ by subsampling cells in $\mathcal{C}_e$ with probability $p_i = \Theta\left( \frac{b (1+\epsilon)^i \log(n)}{\epsilon^3 X}\right)$. Note, this sampling is done agnostically at the start of the stream. 
	\item Run a poly$\left(\frac{\log(n)}{\epsilon}\right)$-sparse recovery algorithm on each substream $\mathcal{S}_i$. For substream $\mathcal{S}_i$, filter out cells $c$ such that $m(c) < (1+\epsilon)^i$. Let $\mathcal{S}'_i$ be the set of cells recovered by the sparse recovery algorithm. Let $\mathcal{S}'_{i| \mathcal{W}_i}$ be the cells in $\mathcal{S}'_i$ that belong to weight class $\mathcal{W}_i$. 
	\item Let $X_{\mathcal{W}'} = \sum_{i \in [b]} X_{\mathcal{W}'_i}$. Let $Z_c$ be a random variable such that $Z_c = \frac{(1+\epsilon)^{i+1}}{p_i}$ if $c \in \mathcal{S}'_{i| \mathcal{W}_i}$ and $0$ otherwise. If $X_{\mathcal{W}'_i} \geq \frac{\epsilon X_{\mathcal{W}'} }{(1+\epsilon)^{i+1}b}$, set the estimator for the $i^{th}$ subsample, $$Y_i = \sum_{c \in \mathcal{S}'_{i| \mathcal{W}_i}}\frac{X_{\mathcal{W}'_i} Z_c}{\left|\mathcal{S}'_{i| \mathcal{W}_i}\right|}$$ 
Otherwise, set $Y_i = 0$. Let $Y_e = \sum_i Y_i$.
    
    %For $c \in \mathcal{S}'_i$ such that $m(c)\leq (1+\epsilon)^{i+1}$, let $Y_i = \sum_{c}\frac{ m(c)}{p_i}$. Else, $Y_i =0$.
    \item Repeat Steps 2-7 for the set $\mathcal{C}_o$ and let $Y_o$ be the corresponding estimator. 
\end{enumerate}
\noindent\textbf{Output:} $Y = \textrm{max}(Y_e, Y_o)$.
\end{Frame}

\subsection{Lower bound for Unit Intervals}

Here, we describe a communication complexity lower bound for estimating $\alpha$ for unit-length interval in turnstile streams and thus show the optimality of Theorem \ref{thm:weighted_ins_del_unit}. Our starting point is the \texttt{Augmented Index} problem and its communication complexity is well understood in the two-player one-way communication model. In this model, we have two players Alice and Bob who are required to compute a function based on their joint input and Alice is allowed to send messages to Bob that are a function of her input and finally Bob announces the answer. Note, Bob isn't allowed to send messages to Alice. 

\begin{definition}(\texttt{Augmented Indexing.})
\label{def:ai}
Let \texttt{AI}$_{n,j}$ denote the communication problem where Alice receives as input $x \in \{0, 1\}^n$ and Bob receives an index $j \in [n]$, along with the $x_{j'}$ for $j' > j$. The objective is for Bob to output $x_j$ in the one-way communication model. 
\end{definition}

\begin{theorem}(Communication Complexity of \texttt{AI}$_{n,j}$, \cite{DBLP:MNSW98}.)
\label{athm:aug_index}
The randomized one-way communication complexity of \texttt{AI}$_{n,j}$ with error probability at most $1/3$ is $\Omega(n)$.
\end{theorem}

Let \texttt{Alg} be a one-pass turnstile streaming algorithm that estimates $\alpha$. We show that \texttt{Alg} can be used as a subroutine to solve \texttt{AI}$_{n,j}$, in turn implying a lower bound on the space complexity of \texttt{Alg}. We formalize this idea in the following theorem:

\begin{theorem}
\label{athm:ins_del_unit_lowerbound}
Any randomized one-pass turnstile streaming algorithm \texttt{Alg} which approximates $\alpha$ to within a $(2-\epsilon)$-factor, for any $\epsilon > 0$, for unit intervals, with at least constant probability, requires $\Omega(n)$ space. 
\end{theorem}
\begin{proof}
Given her input $x$, Alice constructs a stream of unit-length intervals and runs \texttt{Alg} on the stream. For $i \in [n]$, Alice inserts the interval $\left[ \frac{2i - x_i}{n^2} , (\frac{2i - x_i}{n^2}) + 1\right]$. She then communicates the state of \texttt{Alg} to Bob. Bob uses the message received from Alice as the initial state of the algorithm and continues the stream. Since Bob's input includes an index $j$ and $x_{i}$ for all $i > j$, Bob deletes all intervals corresponding to such $i$. Bob then inserts $\left[ (\frac{2j - 0.5}{n^2}) -1 , \frac{2j - 0.5}{n^2} \right]$. 

Let us consider the case where $x_j =1$. We first note that Bob's interval is the leftmost interval in the remaining set. The right endpoint of this interval is $\frac{2j - 0.5}{n^2}$. Next, the rightmost interval corresponds to the $j^{th}$ interval inserted by Alice. The left endpoint of this interval is $\frac{2j -1}{n^2}$. Clearly, these intervals intersect each other and intersect all the intervals between them. Therefore, $\alpha =1$. 

Let us now consider the case where $x_j =0$. Again, Bob's interval is the leftmost with its right endpoint at $\frac{2j - 0.5}{n^2}$. However, the left endpoint of Alice's rightmost interval is $\frac{2j}{n^2}$ and thus these two intervals are independent. Therefore, $\alpha \geq 2$. Observe, any $(2-\epsilon)$-approximate algorithm can distinguish between these two cases and solve \texttt{AI}$_{n,j}$. By Theorem \ref{athm:aug_index}, any such algorithm requires $\Omega(n)$ communication and in turn $\Omega(n)$ space. 
\end{proof}

\section{Arbitrary Length Intervals in Turnstile Streams} 
\label{sec:arb_length}
We now focus on estimating $\alpha$ and $\beta$ for arbitrary-length intervals in turnstile streams. 
%We first rule out any constant factor approximation to $\alpha$ in 
%poly$\left(\frac{\log(n)}{\epsilon}\right)$ space. Note, this also holds for estimating $\beta$. In contrast, Cabello and Perez-Lantero \cite{CP15} showed a $(2+\epsilon)$-approximation to $\alpha$ in insertion only streams using $\poly{\frac{\log(n)}{\epsilon}}$ space. %To find the source of complexity, we study the problem under additional assumptions and parametrization. 
While we cannot obtain streaming algorithms in general, we show it is possible to estimate $\alpha$ and $\beta$ when the maximum 
degree of the interval intersection graph or the maximum length of an
interval arre bounded.  In particular, we show an algorithm that achieves a $(1+\epsilon)$-approximation to $\alpha$
given that the maximum degree is upper bounded by poly$\left(\frac{\log(n)}
{\epsilon}\right)$. We also parameterize the problem with respect to the 
maximum length of an interval, $W_{max}$ (assuming the minimum length is $1$), and give an algorithm using poly$\left(W_{max} \frac{\log(n)}
{\epsilon}\right)$ space.

\subsection{Algorithms under Bounded Degree Assumptions }
\label{subsection:bounded_degree}

In light of the lower bound, we study the problem of estimating $\alpha$ for 
arbitrary-length intervals assuming the number of pair-wise intersections are 
bounded by $\kappa_{\max} =$ poly$\left(\frac{\log(n)}{\epsilon}\right)$. In 
this section we show the following theorem:

\begin{theorem}
\label{thm:bounded_deg_ins_del}
Let $\mathcal{P}$ be an turnstile stream of unit-weight arbitrary-length intervals with lengths polynomially bounded in $n$ and let $\epsilon \in (0, 1/2)$. Let $\kappa_{\max} =$  poly$\left(\frac{\log(n)}{\epsilon}\right)$ be the maximum number of pairwise intersections in $\mathcal{P}$. Then, there exists an algorithm that outputs an estimator $Y$ such that the following guarantees hold:
\begin{enumerate}
    \item $\frac{\alpha}{(1+\epsilon)} \leq Y \leq \alpha$ with probability at least $2/3$.
    \item The total space used is $\text{poly}\left(\frac{\log(n)}{\epsilon}\right)$. 
\end{enumerate}
\end{theorem}

\begin{Frame}[\textbf{Algorithm \ref{alg:level_estimator} : Level Estimator.}]
\label{alg:level_estimator}
\textbf{Input:} Given a turnstile stream $\mathcal{P}$ with unit weight arbitrary length intervals, where the length is polynomially bounded, $\epsilon > 0$ and $\delta >0$, the algorithm outputs a $(1 + \epsilon)$-approximation to $\alpha$, assuming that $\kappa_{\max} =\poly{\frac{log(n)}{\epsilon}}$.

\begin{enumerate}
	\item Let $t =O\left(\frac{\log(n)}{\epsilon}\right)$ be the number of level-classes. Let $\Delta = \bigcup_{i\in[t]} \Delta_i$ be a randomly shifted \textsf{Nested Grid}, where $\Delta_i$ is a grid of side length $\frac{(1+\epsilon)^{i+1}}{\epsilon}$. 
     \item For $i \in [t]$, let $\mathcal{R}_i$ be the set of all $r_i$-\textsf{Structures} at level $i$, where a $r_i$-\textsf{Structure} is a subset of the \textsf{Nested Grid}, $\Delta$, such that there exists an interval at the $i^{th}$ level of the structure, there exist no intervals in the structure at any level $i' > i$ and all the intervals in the structure at levels $i'< i$ intersect the interval at the $i^{th}$ level. 
	\item For all $i \in [t]$, using Algorithm \ref{alg:r_i_structures}, sample poly$\left(\frac{\log(n)}{\epsilon}\right)$ $r_i$-\textsf{Structures} from the set $\mathcal{R}_i$ to create a substream $R^s_i$.  Note, this sampling is carried out with probability $p_i$ defined below. 
    %\item Run a $\kappa_{max}$-\textsf{Sparse Recovery} algorithm on each structure in $R^s_i$. 
    \item At the end of the stream, we recover $R^s_i$, for all $i \in [t]$. Let $Y_i = \frac{|\text{\texttt{OPT}}_{R^s_i}|}{p_i}$ (where $p_i$ is the sampling probability for the $i^{th}$ level), where $|\text{\texttt{OPT}}_{R^s_i}|$ can be computed using an offline algorithm.
\end{enumerate}

\noindent\textbf{Output:} $Y = \sum_{i\in[t]} Y_i$.
\end{Frame}

Let $W$ be the maximum length of the intervals in our input. We split our input into $t = O\left(\frac{\log(n)}{\epsilon}\right)$ length classes $\mathcal{W}_i$ such that for all $i \in [t]$, $\mathcal{W}_i = \{ D_j \in \mathcal{P} | (1+\epsilon)^i \leq r_j \leq (1+\epsilon)^{i+1} \}$. Let $\mathcal{W}$ denote $\bigcup_{i\in[t]} \mathcal{W}_i$. We note that the partition here is over the input to the problem. 

We can estimate the number of non-empty cells in each weight class up to a $(1\pm \epsilon)$-factor by creating a substream for each one and running an $\ell_0$ estimator on them. At the end of the stream, we can discard classes that are not within $\log(W)$ non-empty cells of each other. Therefore, we can assume the remaining classes have the same number of non-empty cells up to a $\log(W)$ factor. 

We then make $O(\log(n))$ guesses for the number of non-empty cells for any fixed level and run our algorithm in parallel for each guess. Since there are $t$ levels, this gives rise to an $O\left(t \log(n)\right)$ factor blowup in space. At the end of the stream we know the correct value for each level via the $\ell_0$ estimates. Let the number of non-empty cells at every level be denoted by $X_i$.

In contrast with our previous algorithm, we note that placing a grid on the input with side length $1$ no longer suffices since our intervals may now lie in multiple cells. Therefore, we impose a nested grid over the input space:

\begin{definition}(\textsf{Nested Grid}.)
Given a partition $\mathcal{W}$, let grid $\Delta_i$, corresponding to $\mathcal{W}_i \in \mathcal{W}$, be a set of cells over the input space with length $\frac{(1+\epsilon)^{i+1}}{\epsilon}$. Then a \textsf{Nested Grid}, denoted by $\Delta$, is defined to be $\bigcup_{i\in[t]} \Delta_i$.
\end{definition}

We then randomly shift the nested grid such that at most an $\epsilon$-fraction of intervals in the $i^{th}$ length class lie within a distance $(1+\epsilon)^{i+1}$ of the $i^{th}$ grid. Since this holds for all $\mathcal{W}_i$, and $\mathcal{W}_i$ are a partition of our input, we lose at most an $\epsilon$-fraction of $\alpha$. We then define the following object that enables us to obtain accurate estimates for each length class. 

\begin{definition}($r_i$-\textsf{Structure.})
\label{def:structure}
We define an $r_i$-\textsf{Structure} to be a subset of the \textsf{Nested Grid}, $\Delta$, such that there exists an interval at the $i^{th}$ level of the structure, there exist no intervals in the structure at any level $i' > i$ and all the intervals in the structure at levels $i'< i$ intersect the interval at the $i^{th}$ level. 
\end{definition}

Let $\mathcal{R}_i$ denote the set of all $r_i$-\textsf{Structures} at level $i$. Observe that, taking the union over $i \in [t]$ of $\mathcal{R}_i$ gives a partition of the input. Therefore, estimating $|\text{\texttt{OPT}}_{\mathcal{R}_i}|$ separately and summing up the estimates is a good estimator for $\alpha$.

Similar to the algorithm in Section \ref{sec:turnstile_unit_interval} a key tool we use is $k$-\textsf{Sparse Recovery}. Intuitively, we subsample poly$\left(\frac{\log(n)}{\epsilon}\right)$ $r_i$-\textsf{Structures} from the set $\mathcal{R}_i$ to create a substream $\mathcal{R}^s_i$ and run a $\kappa_{\textrm{max}}$-\textsf{Sparse Recovery} Algorithm on each substream.  
At the end of the stream, we get an estimate of $|\text{\texttt{OPT}}_{\mathcal{R}_i}|$ that concentrates. We then add up the estimates across all the levels to form our overall estimate. 
We formally describe the Level Estimator Algorithm in Algorithm \ref{alg:level_estimator}, assuming we are given access to a black-box sampling algorithm for sampling an $r_i$-\textsf{Structure}.

Next, we describe the algorithm for sampling $r_i$-\textsf{Structures} in \emph{turnstile streams}. We assume $X_i$ is a $2$-approximation to the number of non-empty cells in $\mathcal{W}_i$. Intuitively, at each level we sample with probability $p_i = \text{poly}\left(\frac{\log(n)}{\epsilon}\right) \frac{1}{X_i}$ and hash each sampled structure such that the structures from $\mathcal{R}_i$ get hashed into different bins.
Simultaneously, the number of structures not in $\mathcal{W}_i$ that get sampled are not too large. Then, given enough bins, we hash each structure into separate bins. Finally, running a sparse recovery on each bin suffices to recover a structure exactly. We show that we can recover enough structures from each set $\mathcal{R}_i$ to get an estimator that concentrates around $\text{\texttt{OPT}}_{R^s_i}$.
We say that set $\mathcal{R}_i$ contributes if $|\mathcal{R}_i| \geq \frac{\epsilon \alpha}{\log(n)}$. Note, if we discard all the sets that do not contribute we lose at most an $\epsilon$-fraction of $\alpha$. We formally describe the sampling algorithm in Algorithm \ref{alg:r_i_structures}.

\begin{lemma}
\label{lem:bounded_unit_weight_arb_length_algorithm}
If the estimator $X_i$ corresponding to the set $\mathcal{R}_i$ passes the threshold, i.e. $X_i > \frac{\epsilon \sum_{i \in [t]} X_i}{t}$, we obtain an estimate $Y_{i}$ such that  $Y_i = (1 \pm \epsilon) |\text{\texttt{OPT}}_{\mathcal{R}_i}|$. If class $\mathcal{R}_i$ does not contribute, we obtain an estimate $Y_i$ such that $Y_i \leq (1+\epsilon) |\text{\texttt{OPT}}_{\mathcal{R}_i}|$.
\end{lemma}
\begin{proof}
We first consider the case where $\mathcal{R}_i$ contributes, i.e., $|\mathcal{R}_i| \geq \frac{\epsilon \alpha}{\log(n)}$. Then, sampling at a rate $p_i$ implies at least $\frac{\log(n)}{\epsilon^2}$ $r_i$-\textsf{Structures} from $\mathcal{R}_i$ survive. We note that since the number of non-empty cells across levels are within an $O\left(\log(n)\right)$ factor of each other, at most $O\left(\frac{\log^2(n)}{\epsilon^3}\right)$ $r_{i'}$-\textsf{Structures} from $R_{i'}$ survive. Note, overall the number of structures that survive is poly$\left(\frac{\log(n)}{\epsilon}\right)$. Therefore, we can run $\kappa_{max}$-\emph{Sparse Recovery} on each cell and stay within our space bounds. At the end of the stream, we can exactly recover the set $R^s_i$ for all $i \in [t]$. For each structure in $R^s_i$, we recover the exact intervals inside them and run an offline algorithm to compute $|\text{\texttt{OPT}}_{\mathcal{R}_i}|$ exactly. Since the size of our sample $R^s_i$ is at least poly$\left(\frac{\log(n)}{\epsilon}\right)$ and an $r_{i'}$-\textsf{Structure} has at most $\kappa_{max}$ independent intervals, we obtain a $(1 \pm \epsilon)$-approximation to $|\text{\texttt{OPT}}_{\mathcal{R}_i}|$ using a simple application of a Chernoff bound. 
In the case where $\mathcal{R}_i$ does not contribute, the number of samples we get is smaller than $\frac{\epsilon \alpha}{\log(n)}$. Therefore, we can set the estimate $Y_i$ for $\mathcal{R}_i$ to be $0$. Note, for such $i$ we do not obtain a concentrated estimate, but we also do not overcount. 
\end{proof}

\begin{Frame}[\textbf{Algorithm \ref{alg:r_i_structures} : Sampling $r_i$-\textsf{Structures} from $\mathcal{R}_i$: .}]
\label{alg:r_i_structures}
\textbf{Input:} Given a turnstile stream $\mathcal{P}$ with unit weight arbitrary length intervals, with the length being polynomially bounded, $\epsilon > 0$ and $\delta >0$, the sampling procedure creates a $\text{poly}\left(\frac{\log(n)}{\epsilon}\right)$ size sample of the set $\mathcal{R}_i$. 

\begin{enumerate}
	\item Let $\Delta_i$ be the $i^{th}$ level of a randomly shifted \emph{Nested Grid} $\Delta$. Let $\mathcal{R}_i$ be the set of $r_i$-\textsf{Structures} where the topmost cells lie in $\Delta_i$. Let $X_i$ be the correct guess for the number of non-empty cells in $\Delta_i$ up to a constant. 
	\item Agnostically sample cells from $\Delta_i$ with probability $p_i= \textrm{max}\left(\text{poly}\left(\frac{\log(n)}{\epsilon}\right) \frac{1}{X_i}, 1\right)$. Let $S_i$ be the corresponding substream created.  
	\item For each cell $c \in \mathcal{S}_i$, let $r^c_i$ be a structure (as defined in \ref{def:structure}) with $c$ at the topmost level. Run $\kappa_{max}$-\emph{Sparse Recovery} on substream $\mathcal{S}_i$.
	\item At the end of the stream, verify that $r^c_i$ is a valid $r_i$-\textsf{Structure}. Let $R^s_i$ be the set of all such structures. 
    \item If $ X_i > \frac{\epsilon \sum_{i\in [t]} X_i }{t}$, keep $R^s_i$, else discard it. 
\end{enumerate}

\noindent\textbf{Output:} $\bigcup_{i \in [t]} R^s_i$.
\end{Frame}

\begin{lemma}
\label{lem:bounded_deg_ins_del_space}
The space used by Unit Weight Arbitrary Length Interval turnstile Algorithm is $poly\left(\frac{\log(n)}{\epsilon}\right)$. 
\end{lemma}
\begin{proof}
First, we note that we make $O\left(t\log(n)\right) = poly\left(\frac{\log(n)}{\epsilon}\right)$ guesses for the $X_i$'s and run our algorithm for each guess in parallel. There are $O\left(\frac{\log(n)}{\epsilon}\right)$ length classes and as seen in Lemma \ref{lem:bounded_unit_weight_arb_length_algorithm} for each class we sample poly$\left(\frac{\log(n)}{\epsilon}\right)$ cells. We run $\kappa_{max}$-\emph{Sparse Recovery} on each sampled cell which requires an additional $poly\left(\frac{\log(n)}{\epsilon}\right)$ space. Therefore, the total space we use is poly$\left(\frac{\log(n)}{\epsilon}\right)$. 
\end{proof}

The proof of Theorem \ref{thm:bounded_deg_ins_del} follows directly from Lemma \ref{lem:bounded_deg_ins_del_space} and Lemma \ref{lem:bounded_unit_weight_arb_length_algorithm}.

\subsection{Algorithms with Parametrized Space Complexity}

In this section we consider the problem of estimating $\alpha$ for arbitrary-length intervals assuming that the space available is at most poly$\left(\frac{W_{max} \log(n)}{\epsilon}\right)$, where $W_{max}$ is an upper bound on the ratio of the max to the min length of an interval. We note that this regime is interesting when $W_{max}$ is sublinear in $n$.    

We begin by modifying the \emph{Nested Grid} $\Delta$, changing the length of a cell at level $i$ to $\frac{(1+ \epsilon)^{i+1}}{2}$. Therefore, randomly shifting the grid, in expectation half the intervals from length class $\mathcal{W}_i$ exactly fit in grid cells $\Delta_i$. Therefore, discarding all the intervals that intersect the cell boundary, we lose at most $1/2$ of the MIS. 

We use the Level Estimators and the Sampling $r_i$-\textsf{Structures} algorithms as described in the previous section and mention the minor modifications that are required. Since we are constrained to $\kappa =$ poly$\left(\frac{\mathcal{W}_{max} \log(n)}{\epsilon}\right)$ space, for each $r_i$-\textsf{Structure} we sample, we run a $\kappa$-\emph{Sparse Recovery} algorithm on it. Observe, a structure can now have $O(n)$ intervals fall in it. Therefore, we maintain an $\ell_0$-estimator that counts the number of non-empty cells in each $r_i$-\textsf{Structure}. If at the end of the stream, the $\ell_0$ estimate is greater than $\kappa$ we discard the $r_i$-\textsf{Structure}.  

Additionally, for each $r_i$-\textsf{Structure} we keep track of the cells in $\Delta_{i'}$ for $i'>i$, s.t they lie vertically above the structure. Note this is an additional $O\left(\mathcal{W}_{max}\log(\mathcal{W}_{max})\right)$ factor. The rest of the algorithm is identical to the bounded degree case. Given that we lose a factor of $2$ during the random shift, repeating the previous analysis, we can estimate $\sum_i \sum_{c \in \Delta_i} |\texttt{OPT}_c| $ to a $(1\ \pm \epsilon)$ factor. Therefore, we obtain an overall $(2+\epsilon)$-approximation to $\alpha$. Further, poly$\left(\frac{\mathcal{W}_{max} \log(n)}{\epsilon}\right)$ space suffices and we obtain the following theorem:

\begin{theorem}
\label{thm:parametrized_space_ins_del}
Let $\mathcal{P}$ be an turnstile stream of unit-weight arbitrary-length intervals s.t. the length is polynomially bounded in $n$ and let $\epsilon \in (0, 1/2)$. Let $\mathcal{W}_{max}$ be an upper bound on the ratio of the max to the min length of intervals in $\mathcal{P}$. Then, there exists an algorithm that outputs an estimator $Y$ s.t. the following guarantees hold:
\begin{enumerate}
    \item $\frac{\alpha}{(2+\epsilon)} \leq Y \leq \alpha$ with probability at least $2/3$.
    \item The total space used is $\text{poly}\left(\frac{\mathcal{W}_{max}\log(n)}{\epsilon}\right)$. 
\end{enumerate}
\end{theorem}

%
%We note that this algorithm also holds in the insertion-only model. Further, making modifications to the level estimator and sampling algorithms, we can get a $(2+\epsilon)$-approximation to $\alpha$ for arbitrary-length intervals assuming that the space available is at most poly$(\frac{W_{max} \log(n)}{\epsilon})$, where $W_{max}$ is an upperbound on the length of an interval (assuming interval lengths are at least $1$). The details of the modifications are deferred to the Appendix. 
%
%\begin{corollary}
%\label{thm:parametrized_space_ins_del}
%Let $P$ be a turnstile stream of unit-weight arbitrary-length intervals such that the length is polynomially bounded in $n$ and let $\epsilon \in (0, 1/2)$. Let $W_{max}$ be an upper bound on the length of intervals in $P$ and $1$ be a lower bound on the length of intervals in $P$. Then, there exists an algorithm that outputs an estimator $Y$ such that the following guarantees hold:
%\begin{enumerate}
%    \item $\frac{\alpha}{(2+\epsilon)} \leq Y \leq \alpha$ with probability at least $2/3$.
%    \item The total space used is $\text{poly}(\frac{W_{max}\log(n)}{\epsilon})$. 
%\end{enumerate}
%\end{corollary}

\section{Unit Radius Disks in Turnstile Streams}
In this section, we present an algorithm to approximate $\alpha$ and $\beta$ for unit-radius disks in $\mathbb{R}^{2}$ that are received in a turnstile stream. We begin with describing an algorithm that achieves a $\left( \frac{8 \sqrt{3}}{\pi} + \epsilon \right)$-approximation to $\alpha$ for unit-radius disks in $\poly{\frac{\log(n)}{\epsilon}}$ space. 
%hen, we show how to extend our algorithm to obtain the same approximation for $\beta$. 
The main algorithmic result we prove is the following:

\begin{theorem}
\label{thm:turnstile_geometric_disks}
Let $\mathcal{P}$ be a sequence of unit-radius disks that are received as a turnstile stream and let $\epsilon \in (0, 1/2)$. Then, there exists an algorithm that outputs an estimator $Y$ such that with probability at least $9/10$
$$ \left(\frac{\pi}{8\sqrt{3}} + \epsilon\right)\beta \leq Y \leq \beta$$
where $\alpha$ is the cardinality of the largest independent set in $\mathcal{P}$. Further, the total space used is $O\left( \textnormal{poly}\left(\frac{\log n}{\epsilon}\right) \right)$.
\end{theorem}

%We note that a greedy algorithm for unweighted disks obtains a $5$-approximation to $\alpha$ \cite{erlebach2003maximum} and the space required is $O\left(\alpha\right)$. The greedy algorithm can be extended to obtain a $(5 + \epsilon)$-approximation in poly$\left(\frac{\log n }{\epsilon}\right)$ space using the sampling approach we presented in Section \ref{sec:turnstile_unit_interval}. However, beating the approximation ratio achieved by the greedy algorithm requires geometric insight.
Critically, we use the hexagonal packing of unit circles in a plane introduced by Lagrange \footnote{See \url{https://en.wikipedia.org/wiki/Circle_packing}}, which was shown to be optimal by Toth \cite{chang2010simple}. The hexagonal packing covers a $\frac{\pi}{\sqrt{12}}$ fraction of the area in two dimensions. We then partition the unit circles in the hexagonal packing into equivalence classes such that two circles in the same equivalence class are at least a unit distance apart. Formally, let $c_1, c_2$ be two unit circles in the hexagonal packing of the plane lying in the same equivalence class. Then, for all points $p_1 \in c_1$, $p_i \in c_2$, $\|p_1 - p_2 \|_2 \geq 1$. Therefore, if input two disks of unit radius have centers lying in distinct circles belong to the same equivalence class, the disks must be independent, as long as the centers do not lie on the boundary. Randomly shifting the underlying hexagonal packing ensures this happens with probability $1$. We then show that we can partition the hexagonal packing into four equivalence classes such that their union covers all the circles in the packing.
%Further, the space required is proportional to the size of the optimal solution, $\alpha$. We then extend the sampling techniques developed unit-length intervals to convert this algorithm into one that runs in $\poly{\frac{\log(n)}{\epsilon}}$ space.

Algorithmically, we first impose a hexagonal grid of circles, $\Delta$, corresponding with side length $1$ and shift it by a random integer. We discard all disks that do not have centers lying inside the grid $\Delta$.   
Given that a hexagonal packing covers a $\frac{\pi}{\sqrt{12}}$ fraction of the area, in expectation, we a discard a $\left( 1 - \frac{\pi}{\sqrt{12}}\right)$ fraction of $\beta$. We note that if we could accurately estimate the remaining \texttt{WMIS}, and scale the estimator by $\frac{\sqrt{12}}{\pi}$, we would obtain a $\left(\frac{\sqrt{12}}{\pi}\right)$-approximation to $\beta$. Let $|\texttt{OPT}_{\textrm{hp}}|$ denote the remaining \texttt{WMIS}.  However, by Theorem \ref{athm:ins_del_unit_lowerbound} such an approximation requires $\Omega(n)$ space.

We then observe that the hexagonal circle packing grid can be partitioned in to four equivalence classes. We use $\mathcal{C}_1, \mathcal{C}_2, \mathcal{C}_3$ and $\mathcal{C}_4$ to denote these equivalence classes. Since the equivalence classes form a partition of the 2-d plane, at least one of them must contain $1/4$-fraction of the remaining maximum independent set. W.l.o.g, let $\mathcal{C}_1$ be the partition that contributes the most to $\beta$. Then, $|\texttt{OPT}_{\mathcal{C}_1}| \geq \frac{1}{4} |\texttt{OPT}_{\textrm{hp}}|$.
Therefore, w.l.o.g. we focus on designing an estimator for $\mathcal{C}_1$. We show a $(1+\epsilon)$-approximation to $\mathcal{C}_1$ in $\poly{\frac{\log(n)}{\epsilon}}$ space. This implies an overall $\left(\frac{4 \sqrt{12}}{\pi} + \epsilon\right) = \left(\frac{8 \sqrt{3}}{\pi} + \epsilon \right)$ approximation for $\beta$.   

\begin{Frame}[Algorithm \ref{alg:naive_approximation_disks} : \textbf{Na\"ive Approximation for Disks.}]
\label{alg:naive_approximation_disks}
\textbf{Input:} Given an turnstile stream $\mathcal{P}$ with weighted unit disks, where the weights are polynomially bounded, $\epsilon$, output a $\left(\frac{36\sqrt{3}}{\pi} + \epsilon\right)$-approximation to $\beta$ with probability $99/100$.
\begin{enumerate}
	\item Let $\Delta$ be a grid of unit radius circles in $\mathbb{R}^2$ arranged as the optimal hexagonal packing. Randomly shift $\Delta$ by $(\alpha, \beta)$, where $\alpha, \beta \in \mathcal{U}(0, \poly{n})$. Partition the cells into equivalence classes, $\mathcal{C}_1, \mathcal{C}_2, \mathcal{C}_3$ and $\mathcal{C}_4$ such that disks lying in distinct cells belonging to the same equivalence class do not intersect. (Note, this can be done for the hexagonal packing of circles.)
	\item Consider a partition of cells in $\mathcal{C}_1$ into $b = \textrm{poly}(\log(n))$ weight classes $\mathcal{W}_i = \{ c \in \mathcal{C}_1 | (1+1/2)^i \leq m(c) < (1+1/2)^{i+1} \}$, where $m(c)$ is the maximum weight of an disk in $c$ (this is not an algorithmic step since we do not know this partition a priori).
    \item Create a substream for each weight class $\mathcal{W}_i$ denoted by $\mathcal{W}'_i$. 
    \item Feed disk $D(d_j,1/2, w_j)$ along substream $\mathcal{W}'_i$ if $w_j \in [(1+1/2)^i, (1+1/2)^{i+1})$.
	\item For each substream $\mathcal{W}'_i$, maintain a $(1\pm\epsilon)$-approximate $\ell_0$-estimator. 
	\item Let $t_i$ be the $\ell_0$ estimate corresponding to $\mathcal{W}'_i$. Let $X_1 = \frac{2}{9(1+\epsilon)} \sum_{i\in [b]} (1+1/2)^{i+1} t_i$.
    \item Repeat Steps 2-6 for the remaining equivalence classes, $\mathcal{C}_2, \mathcal{C}_3$ and $\mathcal{C}_4$ to obtain the corresponding estimator $X_2, X_3, X_4$. 
\end{enumerate}
\textbf{Output:} max$(X_1, X_2, X_3, X_4)$
\end{Frame}

\begin{lemma}
Let $\Delta$ be the hexagonal packing of circles in the plane. Then, $\Delta$ there exists a partitioning of $\Delta$ in to four equivalence classes such that the distance between distinct circles in the same equivalence class is at least $1$.
\end{lemma}

Let $c \in \mathcal{C}_1$ denote a square cell that belongs to the first equivalence class. 
Since we randomly shifted out grid, with probability $1$, no disk has a center that lies on the
boundary. We observe that all disks that lie within cell $c$ must intersect and thus only one 
disk contributes the maximum independent set.
We then snap each disk to the cell containing the center of the disk. We then describe an estimator that gives a $(1+\epsilon)$-approximation to 
$|\text{\texttt{OPT}}_{\mathcal{C}_k}|$ for all $k \in [4]$. 
Therefore, taking the max of the four estimators, we obtain a $(4+\epsilon)$-approximation to 
$\beta$.

Having reduced the problem to estimating $|\text{\texttt{OPT}}_{\mathcal{C}_1}|$, we observe that each even cell has at most $1$ disk that contributed to $\text{\texttt{OPT}}_{\mathcal{C}_1}$, namely the max weight disk landing in the cell. Then, partitioning the cells in $\mathcal{C}_1$ into poly$(\log(n))$ weight classes based on the max weight disk in each cell and approximately counting the number of cells in each weight class suffices to estimate $|\text{\texttt{OPT}}_{\mathcal{C}_1}|$ up to a $(1+\epsilon)$-factor. Given such a partition, we can create a substream for each weight class in the partition and compute the $\ell_0$ norm of each substream. However, we do not know the partition of the cells into the weight classes a priori and this partition can vary drastically over the course of stream given that disks can be deleted. As before, the main technical challenge is to simulate this partition.

We begin by describing a simple algorithm which obtains a $(9/2+\epsilon)$-approximation to $|\text{\texttt{OPT}}_{\mathcal{C}_1}|$ and in turn a $\left(\frac{36\sqrt{3}}{\pi} + \epsilon\right)$-approximation to $\beta$. This estimator is the one introduced in Algorithm \ref{alg:naive_approximation}. Formally, consider a partition of cells in $\mathcal{C}_1$ into $b = \textrm{poly}(\log(n))$ weight classes $\mathcal{W}_i = \{ c \in \mathcal{C}_1 | (1+1/2)^i \leq m(c) < (1+1/2)^{i+1} \}$, where $m(c)$ is the maximum weight of an disk in $c$. Create a substream for each weight class $\mathcal{W}_i$, denoted by $\mathcal{W}'_i$, and feed a disk into this substream if it's weight lies in the range $[(1+1/2)^i, (1+1/2)^{i+1})$. Let $t_i$ be the corresponding $\ell_0$ estimate for substream $\mathcal{W}'_i$. Then, we can approximate the contribution of $\mathcal{W}_i$ by $(1 + 1/2)^{i+1}\cdot t_i$. Summing over the $b$ weight classes gives an estimate for $|\text{\texttt{OPT}}_{\mathcal{C}_1}|$. Given access to an algorithm for estimating the $\ell_0$-norm, Algorithm \ref{alg:naive_approximation_disks} satisfies the following guarantee: 
\begin{lemma}
\label{lem:NaiveApprox_disks}
Algorithm \ref{alg:naive_approximation_disks} outputs an estimate $X$ such that with probability $99/100$, $\left(\frac{36\sqrt{3}}{\pi} + \epsilon\right)\beta \leq X \leq \beta$ and runs in space poly$\left(\frac{\log(n}{\epsilon}\right)$. 
\end{lemma}
\begin{proof}
We observe that for each non-empty cell $c \in \mathcal{C}_1$, there is exactly $1$ disk that can contribute to $|\text{\texttt{OPT}}_{\mathcal{C}_1}|$ since each cell of the grid has side length $1$ and all disks falling in a given cell pairwise intersect. This contributing disk lies in some weight class $\mathcal{W}_i$ and our estimator approximates its weight as $(1+1/2)^{i+1}$. Here, the weights of the disks are sandwiched between $(1+1/2)^i$ and $(1+1/2)^{i+1}$. Therefore, we overestimate the weight by a factor of at most $3/2$. 

Further, instead of taking the maximum over each cell $c$, in the worst case, we may have inserted disks that lie in $c$ into all substreams $\mathcal{W}'_i$, as opposed to only the maximum one. Therefore, we take the sum of our geometrically increasing weight classes over that cell, instead of the maximum weight. In the worst case, we approximate the true weight of a contributing disk, $(3/2)^{i+1}$, with $\sum^{i}_{i'=1}(3/2)^{i'+1} = 3((3/2)^{i+1}-1)$. Note, we again overestimate the weight, this time by a factor of $3$. 

Next, Theorem \ref{thm:ell_0} overestimates the $\ell_0$-norm of $\mathcal{W}_i$ by at most $1+\epsilon$ with probability at least $2/3$. We boost this probability by running $O(\log(n))$ estimators and taking the median. Union bounding over all $i \in [b]$, we simultaneously overestimate the $\ell_0$-norm of all $\mathcal{W}_i$ by at most $1+\epsilon$ with probability at least $99/100$. Therefore, the overall estimator is a $(9/2+\epsilon)$-approximation to $|\text{\texttt{OPT}}_{\mathcal{C}_1}|$. Rescaling our estimator by the above constant underestimates $|\text{\texttt{OPT}}_{\mathcal{C}_1}|$. 

For $i \in [n]$ let $Z_i$ be an indicator random variable that is $1$ if $D(r_{i}, d_i, w_i)$ is centered at a point that lies in the hexagonal circle packing. Let $Z = \sum_{i: D_{i} \in \texttt{OPT}_{\mathcal{P}}} Z_i$. Since $\prob{}{Z_i=1} = \frac{\sqrt{12}}{\pi}$, $\expec{}{Z} = \frac{\pi}{\sqrt{12}}\beta$
%and $\var{Z} = \frac{\pi}{\sqrt{12}}\left(1 - \frac{\pi}{\sqrt{12}}\right)\beta$. 
Then, by Chernoff 
\[
\prob{}{Z \leq (1-\epsilon) \frac{\pi}{\sqrt{12}}\beta}  \leq \exp{\left(-\frac{\pi \epsilon^2 \beta}{3\sqrt{12}}\right)}
%\frac{\frac{\pi}{\sqrt{12}}\left(1 - \frac{\pi}{\sqrt{12}}\right)\beta}{ \left(\epsilon \frac{\pi}{\sqrt{12}}\beta\right)^2} \leq \frac{ \left(1 - \frac{\pi}{\sqrt{12}}\right)}{\epsilon^2 \beta} \leq \frac{1}{\epsilon^2 \beta}
\]
For $\beta = \Omega(\frac{\log(n)}{\epsilon^2})$, we know that $Z \geq (1-\epsilon)\frac{\pi}{\sqrt{12}}\beta$ with probability $1 - \frac{1}{\textrm{poly}(n)}$ and therefore $(1-\epsilon)\frac{\sqrt{12}}{\pi}$ fraction of $\beta$ remains. 

Finally, w.l.o.g, $|\text{\texttt{OPT}}_{\mathcal{C}_1}|\geq \frac{\pi}{4\sqrt{12}} \beta$ and thus  $\left(\frac{36\sqrt{3}}{\pi} + \epsilon\right)\beta \leq X \leq \beta$. 
Since our weights are polynomially bounded, we create poly$\left(\log_{1+\epsilon}(n)\right)$ substreams and run a $\ell_0$ estimator from Theorem \ref{thm:ell_0} on each substream. Therefore, the total space used by Algorithm \ref{alg:naive_approximation_disks} is poly$(\log(n), \epsilon^{-1})$.
\end{proof}

\begin{Frame}[\textbf{Algorithm \ref{alg:weighted_unit_disk_sampling} : Weighted Unit Disk Turnstile Sampling.}]
\label{alg:weighted_unit_disk_sampling}
\textbf{Input:} Given a turnstile stream $\mathcal{P}$ with weighted unit disks, where the weights are polynomially bounded, $\epsilon$, the sampling procedure outputs a $(2 + \epsilon)$-approximation to $\beta$ with probability $99/100$.
\begin{enumerate}
	\item Let $\Delta$ be a grid of unit radius circles in $\mathbb{R}^2$ arranged as the optimal hexagonal packing. Randomly shift $\Delta$ by $(\alpha, \beta)$, where $\alpha, \beta \in \mathcal{U}(0, \poly{n})$. Partition the cells into equivalence classes, $\mathcal{C}_1, \mathcal{C}_2, \mathcal{C}_3$ and $\mathcal{C}_4$ such that disks lying in distinct cells belonging to the same equivalence class do not intersect.
    \item For cells in $\mathcal{C}_1$, snap each disk in the input to a cell $c$ that contains its center.
	\item Consider a partitioning of the cells in $\mathcal{C}_1$ into $b = \text{poly}\left(\frac{\log(n)}{\epsilon}\right)$ weight classes $\mathcal{W}_i = \{ c \in \mathcal{C}_1 | (1+\epsilon)^i \leq m(c) \leq (1+\epsilon)^{i+1}$ \}, where $m(c)$ is the maximum weight of an disk in $c$ (we do not know this partition a priori.) Create a substream for each weight class $\mathcal{W}_i$ denoted by $\mathcal{W}'_i$. 
    \item Feed disk $D(d_j,1, w_j)$ along substream $\mathcal{W}'_i$ such that $w_j \in [(1+\epsilon)^i, (1+\epsilon)^{i+1})$. Maintain a $(1\pm\epsilon)$-approximate $\ell_0$-estimator for each substream. Let $|\mathcal{W}'_i|$ denote the number of non-empty cells in substream $\mathcal{W}'_i$ and $X_{\mathcal{W}'_i}$ be the corresponding estimate returned by the $\ell_0$-estimator. 
	\item Create $O(\log(n))$ substreams, one for each guess of $|\text{\texttt{OPT}}_{\mathcal{C}_1}|$. Let $X$ be the guess for the current branch of the computation. In parallel, run Algorithm \ref{alg:naive_approximation} estimates $|\text{\texttt{OPT}}_{\mathcal{C}_1}|$ up to a constant factor. Therefore, at the end of the stream, we know a constant factor approximation to the correct value of $|\text{\texttt{OPT}}_{\mathcal{C}_1}|$ and use the estimator from corresponding branch of the computation.   
	\item In parallel, for $i \in [b]$, create substream $\mathcal{S}_i$ by subsampling cells in $\mathcal{C}_1$ with probability $p_i = \Theta\left( \frac{b (1+\epsilon)^i \log(n)}{\epsilon^3 X}\right)$. Note, this sampling is done agnostically at the start of the stream. 
	\item Run a poly$\left(\frac{\log(n)}{\epsilon}\right)$-sparse recovery algorithm on each substream $\mathcal{S}_i$. For substream $\mathcal{S}_i$, filter out cells $c$ such that $m(c) < (1+\epsilon)^i$. Let $\mathcal{S}'_i$ be the set of cells recovered by the sparse recovery algorithm. Let $\mathcal{S}'_{i| \mathcal{W}_i}$ be the cells in $\mathcal{S}'_i$ that belong to weight class $\mathcal{W}_i$. 
	\item Let $X_{\mathcal{W}'} = \sum_{i \in [b]} X_{\mathcal{W}'_i}$. Let $Z_c$ be a random variable such that $Z_c = \frac{(1+\epsilon)^{i+1}}{p_i}$ if $c \in \mathcal{S}'_{i| \mathcal{W}_i}$ and $0$ otherwise. If $X_{\mathcal{W}'_i} \geq \frac{\epsilon X_{\mathcal{W}'} }{(1+\epsilon)^{i+1}b}$, set the estimator for the $i^{th}$ subsample, $Y_i = \sum_{c \in \mathcal{S}'_{i| \mathcal{W}_i}}\frac{X_{\mathcal{W}'_i} Z_c}{\left|\mathcal{S}'_{i| \mathcal{W}_i}\right|}$. 
Otherwise, set $Y_i = 0$. Let $Y_e = \sum_i Y_i$.
    
    %For $c \in \mathcal{S}'_i$ such that $m(c)\leq (1+\epsilon)^{i+1}$, let $Y_i = \sum_{c}\frac{ m(c)}{p_i}$. Else, $Y_i =0$.
    \item Repeat Steps 2-7 for the sets $\mathcal{C}_2, \mathcal{C}_3, \mathcal{C}_4$ and let $Y_2, Y_3, Y_4$ be the corresponding estimators. 
\end{enumerate}
\textbf{Output:} $Y = \textrm{max}(Y_1, Y_2, Y_3, Y_4)$.
\end{Frame}

We can thus assume we know $\beta$ and $|\text{\texttt{OPT}}_{\mathcal{C}_1}|$ up to a constant by initially making $O\left(\log(n)\right)$ guesses and running Algorithm \ref{alg:naive_approximation_disks} for each guess in parallel. Intuitively, similar to the disk case, we simulate partitioning cells in $\mathcal{C}_1$ into $\text{poly}\left(\frac{\log(n)}{\epsilon}\right)$ weight classes according to the maximum weight occurring in each cell. Since we do not know this partition a priori, we initially create $b =O\left(\frac{\log(n)}{\epsilon}\right)$ substreams, one for each weight class and run the $\ell_0$-estimator on each one. We then make $O\left(\frac{\log(n)}{\epsilon}\right)$ guesses for $|\text{\texttt{OPT}}_{\mathcal{C}_1}|$ and run the rest of the algorithm for each branch in parallel. 

Additionally, we run the Algorithm \ref{alg:naive_approximation_disks} to compute the right value of $|\text{\texttt{OPT}}_{\mathcal{C}_1}|$ up to a constant factor, which runs in space $\text{poly}\left(\frac{\log(n)}{\epsilon}\right)$. Then, we create $b = \text{poly}\left(\frac{\log(n)}{\epsilon}\right)$ substreams by agnostically sampling cells with probability $p_i = \Theta\left( \frac{b (1+\epsilon)^i \log(n)}{\epsilon^3 X}\right)$. Sampling at this rate preserves a sufficient number of cells from weight class $\mathcal{W}_i$. We then run a sparse recovery algorithm on the resulting substreams. 

The analysis for estimating the contribution of each substream is the same as in the case of disks. We sketch an outline of the proof here. Observe, the resulting substreams are sparse since we can filter out cells that belong weight classes $\mathcal{W}_{i'}$ for $i' < i$ by simply checking if the maximum disk seen so far lies in weight classes $\mathcal{W}_{i}$ and higher. Further, sampling with probability proportional to $\Theta\left( \frac{b (1+\epsilon)^i \log(n)}{\epsilon^3 |\text{\texttt{OPT}}_{\mathcal{C}_1}|}\right)$ ensures that the number of cells that survive from weight classes $\mathcal{W}_{i}$ and above are small. Therefore, we recover all such cells using the sparse recovery algorithm. Note, we limit the algorithm to considering weight classes that have a non-trivial contribution to $\text{\texttt{OPT}}_{\mathcal{C}_1}$. 

Using the $\ell_0$ norm estimates computed above, we can determine number on non-empty cells in each of the weight classes. Thus, we create a threshold for weight classes that contribute, such that all the weight classes below the threshold together contribute at most an $\epsilon$-fraction of $|\text{\texttt{OPT}}_{\mathcal{C}_1}|$ and we can set their corresponding estimators to $0$. Further, for all the weight classes above the threshold, we can show that sampling at the right rate leads to recovering enough cells to achieve concentration in estimating their contribution.

We observe that the space and correctness analysis for each equivalence class is identical to the $1$-d case in Section \ref{sec:turnstile_unit_interval} since it does not depend on the geometry of the objects that are inserted into substeam $\mathcal{S}_i$. Theorem \ref{thm:turnstile_geometric_disks} follows.

%Next, we show that the total space used by Algorithm \ref{alg:weighted_unit_disk_sampling} is $\text{poly}\left(\frac{\log(n)}{\epsilon}\right)$. We initially create $b =O\left(\frac{\log(n)}{\epsilon}\right)$ substreams, one for each weight class and run a $\ell_0$-estimator on each one. Recall, this requires $\text{poly}\left(\frac{\log(n)}{\epsilon}\right)$. We then make $O\left(\frac{\log(n)}{\epsilon}\right)$ guesses for $|\text{\texttt{OPT}}_{\mathcal{C}_e}|$ and run the rest of the algorithm for each branch in parallel. Additionally, we run Algorithm \ref{alg:naive_approximation} to compute the right value of $|\text{\texttt{OPT}}_{\mathcal{C}_e}|$ up to a constant factor, which runs in space $\text{poly}\left(\frac{\log(n)}{\epsilon}\right)$. Then, we create $b$ substreams by sampling cells with probability $p_i = \Theta\left( \frac{b (1+\epsilon)^i \log(n)}{\epsilon^3 X}\right)$, for $i \in [b]$. Subsequently, we run a $\text{poly}\left(\frac{\log(n)}{\epsilon}\right)$-sparse recovery algorithm on each one. Note, if each sample is not too large, this can be done in $\text{poly}\left(\frac{\log(n)}{\epsilon}\right)$ space. Therefore, it remains to show that each sample $\mathcal{S}_i$ is small. 

\section{Insertion-Only Streams}
\label{sec:ins_only}
In this section, we describe an algorithm that obtains a $\left(\frac{3}{2} + \epsilon\right)$-approximation for estimating the maximum weighted independent set of intervals in $\poly{\frac{log(n)}{\epsilon}}$ space, given that we are not allowed to delete any intervals. Recall, \cite{CP15} show that $\left(\frac{3}{2} + \epsilon\right)$ is tight for the unweighted case in insertion-only streams. We also show a lower bound for estimating the maximum independent set of disks in insertion-only streams. The lower bound for intervals in \cite{CP15} shows that $\left(\frac{3}{2} - \epsilon\right)$-approximation requires $\Omega(n)$ space and this naturally extends to disks. We improve this to $2-\epsilon$, implying a strict separation between intervals and disks for insertion-only streams. Note, this is not yet known to be the case for turnstile streams.

\subsection{Intervals} 
We present a single-pass insertion-only streaming algorithm that approximates $\beta$ for unit-length intervals. We begin with describing an algorithm that achieves a $(\frac{3}{2} + \epsilon)$-approximation to $\beta$ in $O(\beta)$ space. Then, we describe a sampling procedure that creates a sketch of the data structure used by the previous algorithm and show an estimator that outputs a $(\frac{3}{2} + \epsilon)$-approximation to $\beta$. Further, the space used by the sketch is $\text{poly}\left(\frac{\log(n)}{\epsilon}\right)$. 

\begin{theorem}
\label{thm:weighted_ins_unit_int}
Let $P$ be an insertion-only stream of weighted unit intervals s.t. the weights are polynomially bounded in $n$ and let $\epsilon \in (0, 1/2)$. Then, there exists an algorithm that outputs an estimator $Y$ s.t. with probability at least $9/10$ the following guarantees hold:
\begin{enumerate}
    \item $\frac{2 \beta}{3+\epsilon} \leq Y \leq \beta$.
    \item The total space used is $\text{poly}\left(\frac{\log(n)}{\epsilon}\right)$ bits. 
\end{enumerate}
\end{theorem}

We use an algorithm with the following guarantee as a subroutine: 
\begin{lemma}
\label{lem:weighted_unit_ins_only}
Let $\mathcal{P}$ be an insertion-only stream of $n$ weighted unit intervals. Then, the Weighted Unit Interval Selection Algorithm outputs an estimator $Y$ and the following guarantees hold:
\begin{enumerate}
    \item $\frac{2\beta}{3 + \epsilon} \leq Y \leq \beta$
    \item The total space used is $\widetilde{O}\left(\frac{\beta}{\epsilon}\right)$ bits.
\end{enumerate}
\end{lemma}

\begin{Frame}[\textbf{Algorithm \ref{alg:weighted_unit_interval_selection_sampling} : Weighted Unit Interval Insertion-Only Sampling.}]
\label{alg:weighted_unit_interval_selection_sampling}
\textbf{Input:} Given an insertion-only stream $\mathcal{P}$ of weighted unit intervals and $\epsilon >0$, the sampling procedure outputs an estimate $Y$ s.t. it satisfies the guarantees of Theorem \ref{thm:weighted_ins_unit_int}.

\begin{enumerate}
	\item Make $O(\log(n))$ guesses for $\beta$. Let $X$ be the right guess. 
	\item Consider a partitioning $\mathcal{C}_i$, where $\mathcal{C}_i$ is the set of all cells $c$ s.t. the maximum weight of an interval in $c$ is in the range $\mathcal{W}_i$. (we do not explicitly know this partitioning.)
	\item Create sample $\mathcal{S}_i$, corresponding to partition $\mathcal{C}_i$, by sampling non-empty cells in $\Delta$ with probability $p_i = \text{poly}\left(\frac{\epsilon X}{(1 + \epsilon)^i \log(n)}\right)$.
	\item If $c \in \mathcal{S}_i$, discard all intervals in $c$ s.t. their weight is less that $\epsilon^2 (1 + \epsilon)^i$. 
	\item Let $Y^c_i$ be the output of running Weighted Unit Interval Selection on $c \in \mathcal{S}_i$.
	\item Then, for all $i$, $Y_i = \sum_{c \in \mathcal{S}_i} \frac{Y^c_i}{p_i}$. If $Y_i < \text{poly}\left(\frac{\epsilon}{\log(n)}\right)X$, set $Y_i = 0$. 
	
\end{enumerate}

\textbf{Output:} $\sum_i Y_i$.
\end{Frame}

\begin{proof}
We first show that the space bound holds. Note, the number of non-empty cells in $\Delta$ are at most $\beta$. For each non-empty cell $c$ we store at most $2$ intervals per weight class. The total number of weight classes is $O\left(\frac{\log(n)}{\epsilon}\right)$ and thus we store at most $O(\log(n))$ information for each non-empty cell. Overall, this gives a space bound of $\widetilde{O}\left(\frac{\beta}{\epsilon}\right)$. 

The argument for is very similar to the one for the unweighted case by Emek et. al. \cite{EHR12} Consider any cell $c$ of size $r = O\left(1/\epsilon\right)$. Let $D_1, D_2 \ldots D_r$ be \texttt{OPT}$_c$. For any three intervals above, we have stored at least two intervals contained in their union. Then, in expectation we have stored $\frac{2}{3} \text{\texttt{OPT}}_c$, therefore, there exists a set of disjoint intervals that a combined contribution of $\frac{2}{3} |\text{\texttt{OPT}}_c|$. Note, this part of the algorithm is deterministic.   
\end{proof}

Next, we show a sampling procedure that samples $\text{poly}\left(\frac{\log(n)}{\epsilon}\right)$ non-empty cells and maintains the same data structure as the Weighted Unit Interval Selection Algorithm. Then, our overall estimator is sum of the \texttt{MWIS} in the sampled cells, scaled up by the probability of sampling. We place a grid $\Delta$ on the input space of side length $\frac{1}{\epsilon}$. We then randomly shift each input interval and discard any interval that intersects the grid. Note, we therefore lose at most an $\epsilon$-fraction of $\beta$.

We first focus on the space used by our sampling process. Intuitively, we sample $\text{poly}\left(\frac{\log(n)}{\epsilon}\right)$ cells from the grid $\Delta$ and for each cell run Algorithm 3 on it. Since each cell has at most $O(\frac{1}{\epsilon})$ independent intervals, the size of the optimal solution in a cell is at most $O\left(\frac{1}{\epsilon} \log(n)\right)$. Therefore the overall space used is $\text{poly}\left(\frac{\log(n)}{\epsilon}\right)$. To finish the proof for the space complexity of our algorithm it remains to show the following lemma:

\begin{lemma}
Given an insertion only stream $\mathcal{P}$, the Weighted Unit Interval Insertion-Only Sampling procedure samples $\text{poly}\left(\frac{\log(n)}{\epsilon}\right)$ cells from the grid $\Delta$ with high probability.  
\end{lemma}
\begin{proof}
For every cell $c \in \mathcal{C}_i$, we sample it with probability $p_i = \text{poly}\left(\frac{\epsilon X}{(1 + \epsilon)^i \log(n)}\right)$. Note, the cardinality of the set $\mathcal{C}_i$ is at most $\frac{\beta}{(1+\epsilon)^i} = (1 \pm 1/2)\frac{X}{(1+\epsilon)^i}$. Therefore, in expectation we sample $\text{poly}\left(\frac{\log(n)}{\epsilon}\right)$ from $\mathcal{C}_i$. By Chernoff, the sample isn't larger than a constant factor with high probability. Since the number of weight classes is at most $O\left(\frac{\log(n)}{\epsilon}\right)$, the total space used by our algorithm is $O\left(\frac{\log(n)}{\epsilon}\right)$. 
\end{proof}

\begin{Frame}[\textbf{Algorithm \ref{alg:weighted_unit_interval_selection} : Weighted Unit Interval Selection.}]
\label{alg:weighted_unit_interval_selection}
\textbf{Input:} Given a one-pass insertion-only stream $\mathcal{P}$ with weighted unit-intervals, where the weights are polynomially bounded and $\epsilon >0$, the Weighted Unit Interval Selection Algorithm outputs a $(\frac{3}{2} + \epsilon)$-approximation to $\beta$ in $O(\beta)$ space.
\begin{enumerate}
	\item Randomly shift grid $\Delta$ with cells of side length $\frac{1}{\epsilon}$.
	\item Consider $O(\frac{\log(n)}{\epsilon})$ weight classes $\mathcal{W}_i = \{ D_j | (1 + \epsilon)^i \leq w_j < (1+\epsilon)^{i+1} \}$.
	\item Given a cell $c$ and weight class $\mathcal{W}_i$, let $l^{w_i}_c$ and $r^{w_i}_c$ be the left-most left endpoint and right-most right endpoint intervals in $c$ and in weight class $\mathcal{W}_i$. 
	\item For each non-empty cell $c$ in $\Delta$, for each weight class $\mathcal{W}_i$, maintain $l^{w_i}_c$ and $r^{w_i}_c$. 
\end{enumerate}
\textbf{Output:} \texttt{WIS} of the remaining intervals.
\end{Frame}

It remains to show that the estimate returned by our sampling procedure is indeed a $(\frac{3}{2} + \epsilon)$-approximation. We first observe that the union of the $\mathcal{C}_i$'s form a partition of our input space. Therefore, it suffices to show that we obtain a $(\frac{3}{2} + \epsilon)$-approximation to the \texttt{MWIS} for each $\mathcal{C}_i$. Let $c$ denote a cell in $\mathcal{C}_i$ and $\text{\texttt{OPT}}_c$ denote the \texttt{MWIS} in cell $c$. Further, we say class $\mathcal{C}_i$ contributes if $(1+\epsilon)^i|\mathcal{C}_i| \geq \text{poly}\left(\frac{\epsilon}{\log(n)}\right)X$.

\begin{lemma}
If class $\mathcal{C}_i$ contributes, we obtain an estimate $Y_{i}$ s.t.  $Y_i = (1 \pm \epsilon) \sum_{c \in \mathcal{C}_i} |\text{\texttt{OPT}}_c|$ with high probability. If class $\mathcal{C}_i$ does not contribute, we obtain an estimate $Y_i$ s.t. $Y_i \leq (1+\epsilon) \sum_{c \in \mathcal{C}_i} |\text{\texttt{OPT}}_c|$ with high probability. Overall, $\sum_i Y_i = (1\pm\epsilon)\sum_{c\in \Delta} |OPT_c|$ with probability at least $1-1/n$.
\end{lemma}
\begin{proof}
Let the max weight for $c \in \mathcal{C}_i$ be $w$. Then, $w \geq (1+\epsilon)^i$ and $ w \leq |\text{\texttt{OPT}}_c| \leq \frac{(1+\epsilon)^{i+1}}{\epsilon}$. Note, the algorithm ignores all intervals of weight at most $\epsilon^2 w$, we lose at most $\epsilon|\text{\texttt{OPT}}_c|$ since the dropped intervals can contribute at most $\frac{\epsilon^2w}{\epsilon} \leq \epsilon|\text{\texttt{OPT}}_c|$. 

We first consider the case where $\mathcal{C}_i$ contributes. Then, sampling at a rate $p_i$ implies at least $p_i |\mathcal{C}_i|$ cells survive in expectation. By Chernoff,
\[
\Pr\left[|\mathcal{S}_i| \leq (1-\epsilon)(\frac{\log(1/\delta)}{2\epsilon^2})\right] \leq \textrm{exp}\left(\frac{-2\epsilon^2 \log(1/\delta)}{2\epsilon^2}\right) \leq \delta
\]

Setting $1/\delta = \textrm{exp}\left(\text{poly}(\frac{\log(n)}{\epsilon})\right) $ and union bounding over all $i\in O\left(\frac{\log(n)}{\epsilon}\right)$, $|\mathcal{S}_i| = \Omega(\text{poly}(\frac{\log(n)}{\epsilon}))$ with probability at least $1-1/n^c$. We then compute $|\text{\texttt{OPT}}_c|$ and scale it up by $p_i$. Then, our estimator is $r_i = \sum_{c\in \mathcal{C}_i} \frac{|\text{\texttt{OPT}}_c|}{p_i}$. Then, $E[r_i] = \sum_{c \in \mathcal{C}_i}|\text{\texttt{OPT}}_c|$. Since for each cell $c \in \mathcal{C}_i$, $(1+\epsilon)^i \leq |\text{\texttt{OPT}}_c| \leq \frac{(1+\epsilon)^{i+1}}{\epsilon}$ and the number of samples are $\text{poly}\left(\frac{\log(n)}{\epsilon}\right)$, therefore, by Chernoff bounds, our estimate is $(1 \pm \epsilon)\sum_{c \in \mathcal{C}_i}|\text{\texttt{OPT}}_c|$ with probability at least $1-1/n^c$. Therefore, for cells that contribute, our estimator concentrates with probability at least $1-1/2n$.   

In the case where $\mathcal{C}_i$ does not contribute, the number of samples we get is smaller than $\text{poly}\left(\frac{\epsilon}{\log(n)}\frac{X}{(1+\epsilon)^i}\right)$. Since each $|\text{\texttt{OPT}}_c| \leq \frac{(1+\epsilon)^{i+1}}{\epsilon}$, the total contribution of the sample is at most $\text{poly}(\frac{\epsilon}{\log(n)})X$ in expectation. By a Chernoff bound similar to the one above, the sample size $|\mathcal{S}_i|$, for all $\mathcal{C}_i$ that do not contribute, is $\text{poly}\left(\frac{\epsilon}{\log(n)}\frac{X}{(1+\epsilon)^i}\right)$ with probability at least $1-1/2n$. Therefore, our estimate $Y_i < \text{poly}\left(\frac{\epsilon}{\log(n)}\right)X$ and is set to $0$. Note, in this case we do not obtain a concentration, but we also do not over count. Union bounding over the two cases for $\mathcal{C}_i$, the lemma holds with $1-1/n$ probability.
\end{proof}

\textbf{Proof of Theorem \ref{thm:weighted_ins_unit_int}.} We observe that for $c \in \mathcal{C}_i$, we cannot exactly compute $|\text{\texttt{OPT}}_c|$ in a stream. However, we can run Algorithm 3 on each cell to obtain a $\frac{3}{2}$-approximation to $|\text{\texttt{OPT}}_c|$ with at least constant probability. Therefore, overall we obtain an estimate that gives a $\frac{3(1+\epsilon)}{2}$-approximation. 

Note, Lemma \ref{lem:weighted_unit_ins_only} guarantees that Algorithm 3 runs in space $|\text{\texttt{OPT}}_c|$ and since each cell $c$ has length $\frac{1}{\epsilon}$, $|\text{\texttt{OPT}}_c| \leq O(\frac{1}{\epsilon})$. Therefore, the total space used by the sampling procedure is $\text{poly}\left(\frac{\log(n)}{\epsilon}\right)$.

\subsection{Disks}
We describe a lower bound for estimating $\alpha$ for unit disks in insertion-only streams via a reduction from the communication complexity of the \texttt{Indexing} problem, which we use as the starting point. We consider the one-way communication model between two players Alice and Bob and each player has access to private randomness. The randomized communication complexity of \texttt{Indexing} is well understood in the two-player one-way communication model.

\begin{definition}(\texttt{Indexing})
Let \texttt{I}$_{n,j}$ denote the communication problem where Alice receives as input a bit vector $x \in \{0, 1\}^n$ and Bob receives an index $j \in [n]$. The objective is for Bob to output $x_j$ under the one-round one-way communication model with error probability at most $1/3$. 
\end{definition}

\begin{theorem}(Communication Complexity of \texttt{I}$_{n,j}$.)
\label{thm:index}
The randomized one-round one-way communication complexity of \texttt{I}$_{n,j}$ with error probability at most $1/3$ is $\Omega(n)$.
\end{theorem}

We begin with considering the stream of disks $\mathcal{P}$. Let \texttt{Alg} be a one-pass insertion-only streaming algorithm that estimates the cardinality of the maximum independent set denoted by $\alpha$. We then show that \texttt{Alg} can be used as a subroutine to solve the communication problem \texttt{I}$_{n,j}$. Therefore, a lower bound on the communication complexity in turn implies a lower bound on the space complexity of \texttt{Alg}. Formally, 

\begin{theorem}
\label{thm:insertion_only_disks_lowerbound}
Given a stream of disks $\mathcal{P}$, any randomized one-pass insertion-only streaming algorithm \texttt{Alg} which approximates $\alpha$ to within a $(2-\epsilon)$-factor, for any $\epsilon > 0$, with error at most $1/3$, requires $\Omega(n)$ space. 
\end{theorem}
\begin{proof}
We show that any such insertion-only streaming algorithm \texttt{Alg} can be used to construct a randomized protocol $\Pi$ to solve the communication problem.
Given her input $x$, Alice constructs a stream of unit disks and runs \texttt{Alg} on the stream. Consider the unit circle around the origin. Divide the half-circle above the $x$-axis into $n$ equally spaced points, denoted by vectors $p_1, p_2, \ldots, p_n$.  
For $i \in [n]$, if $x_i = 0$, Alice streams a unit disk centered at $p_i$. If $x_i=1$, Alice streams a unit disk centered at $-p_i$.
After streaming $n$ disks, Alice communicates the memory state of \texttt{Alg} to Bob. Bob uses the message received from Alice as the initial state of the algorithm and continues the stream. Recall, Bob's input only consists of a single index $j$. Therefore, Bob inserts a unit disk centered at $(1 + 1/n^2)p_j$. 

We first observe that all disks inserted by Alice pairwise intersect. Since all her unit radius disks are centered on the unit circle around the origin, the distance between their center and the origin is $1$. Since all the disks contain the origin, they pairwise intersect. Now, let us consider the case where $x_j= 0$. Recall, in this case, Alice inserts the disk centered $p_j$ and Bob inserts the disk centered at $(1+1/n^2)p_j$. The distance between their centers is $1/n^2$ and they clearly intersect. Let us now consider the other disks inserted by Alice, centered at points $p_i$ for $i \neq j$. The distance between their centers is 
\begin{equation}
\begin{split}
|| p_i - (1+ 1/n^2)p_j ||^2_2 & = ||p_i ||^2_2 + (1 +1/n^2)^2 || p_j ||^2_2 \pm 2(1+1/n^2) \langle p_i, p_j \rangle \\
& \leq 1 + (1 + 3/n^2) \pm 2(1 + 1/n^2)\langle p_i, p_j \rangle 
\end{split}
\end{equation}
where the last inequality follows from $(1 +1/n^2)^2 = 1 + 1/n^4 +2/n^2 \leq 1 + 3/n^2$ for sufficiently large $n$. Since $i\neq j$, $\langle p_i, p_j \rangle \leq 1 - \Theta(1/n)$. Note, $(1+ 1/n^2)(1 - \Theta(1/n)) \leq 1 - \Theta(n)$ for sufficiently large $n$. Substituting this above, we get 
\begin{equation}
\begin{split}
|| p_i - (1+ 1/n^2)p_j ||^2_2 & \leq 1 + (1 + 3/n^2) \pm 2(1 + 1/n^2)(1 - \Theta(1/n)) \\
& \leq 2 + 3/n^2 \pm 2(1 - \Theta(1/n)) \\
& \leq 4 - \Theta(1/n)
\end{split}
\end{equation}
where the last inequality follows from $\Theta(1/n) \geq 3/n^2$ for sufficiently large $n$. Therefore, the squared distance between the centers is strictly less $4$ and the disks do intersect. As a consequence, all disks pairwise intersect and $\alpha = 1 $.

Let us now consider the case where $x_j =1$. Recall Alice inserts a disk centered at $-p_j$ and Bob inserts a disk centered at $(1 + 1/n^2)p_j$. The distance between the centers is $(2 + 1/n^2)$, therefore the two disks do not intersect. Then, $\alpha$ is at least $2$. We observe that any $(2-\epsilon)$-approximate algorithm \texttt{Alg} can distinguish between these two cases because in the first case \texttt{Alg} outputs $at\ most$ $1$ and in the second case \texttt{Alg} outputs $at\ least$ $1 + \epsilon$. Therefore it is a valid protocol for solving \texttt{I}$_{n,j}$. If \texttt{Alg} has error at most $1/3$, the protocol has error at most $1/3$. 
By Theorem \ref{thm:index}, any such protocol requires $\Omega(n)$ communication and in turn \texttt{Alg} requires $\Omega(n)$ space. 
\end{proof}

\end{document}